\documentclass[11pt]{article}

\RequirePackage{amssymb} 
\RequirePackage{amsmath} 
\RequirePackage{mathrsfs} 
\RequirePackage{latexsym}

\usepackage{booktabs}

\usepackage{adjustbox}

\usepackage{amssymb}

\usepackage{xcolor}

\usepackage{ragged2e}

\usepackage{placeins}

\usepackage{caption}

\usepackage{subcaption}

\usepackage{color}

\usepackage{array}
\newcolumntype{H}{>{\setbox0=\hbox\bgroup}c<{\egroup}@{}}
\newcolumntype{N}{@{}m{0pt}@{}}

\usepackage{lscape}

\usepackage{multirow}

\usepackage{rotating}

\usepackage{bigints}

\usepackage{setspace}

\usepackage{amsmath}

\usepackage{setspace}

\usepackage{bm}

\usepackage{forest}
\usetikzlibrary{arrows.meta,quotes}

\DeclareFontFamily{U}{stixscr}{}
\DeclareFontShape{U}{stixscr}{m}{n}{<-> s*[1.1] stix-mathscr}{}

\newcommand{\R}{\ensuremath{\mathbb{R}}}

\usepackage{bbm}

\makeatletter
\newcommand*{\indep}{%
	\mathbin{%
		\mathpalette{\@indep}{}%
	}%
}
\newcommand*{\nindep}{%
	\mathbin{
		\mathpalette{\@indep}{\not}
	}%
}
\newcommand*{\@indep}[2]{%
	\sbox0{$#1\perp\m@th$}
	\sbox2{$#1=$}
	\sbox4{$#1\vcenter{}$}
	\rlap{\copy0}
	\dimen@=\dimexpr\ht2-\ht4-.2pt\relax
	\kern\dimen@
	{#2}%
	\kern\dimen@
	\copy0 
} 
\makeatother

\usepackage{relsize}

\usepackage{placeins}
\usepackage{graphicx}

\usepackage[round]{natbib}   
\usepackage[american]{babel}
\usepackage{csquotes}

\usepackage{mathtools}

\usepackage{xr-hyper}
\pdfoutput=1
\usepackage{hyperref}
\hypersetup{
	colorlinks=true,
	linkcolor=blue,
	citecolor=blue,
	filecolor=blue,
	urlcolor=blue,
}

\usepackage{amsthm}
\newtheorem{proposition}{Proposition}[section]

\newtheorem{corollary}{Corollary}[section]
\newtheorem{lemma}{Lemma}[section]

\newtheorem{proofprop}{}
\newtheorem{prooftheorem}{}

\newtheorem{remark}{Remark}[section]

\newtheorem{theorem}{Theorem}[section]
\newtheorem{assumption}{Assumption}[section]
\newtheorem{definition}{Definition}[section]

\usepackage{setspace}
\title{\textbf{On the falsification of instrumental variable models for heterogeneous treatment effects}}
\author{
	Ricardo E. Miranda\thanks{Department of Economics, Duke University, \texttt{ricardo.e.miranda@duke.edu}}
	\thanks{I thank Adam Rosen, Michael Pollmann, Matthew Masten, and Arnaud Maurel for their invaluable guidance and support. I also thank participants at the Duke microeconometrics breakfast and ITAM Alumni conference for their useful comments and suggestions. All errors are my own.  }
	\date{January 20, 2026}
}

\begin{document}
	\maketitle
	
	\begin{abstract}
		In this paper I derive a set of testable implications for econometric models defined by three assumptions: (i) the existence of strictly exogenous discrete instruments, (ii) restrictions on how the instruments affect adoption of a finite number of treatment types (such as monotonicity), and (iii) the assumption that the instruments only affect outcomes through their effect on treatment adoption (i.e. an exclusion restriction). The testable implications aggregate (via integration) an otherwise potentially infinite set of inequalities that must hold for every measurable subset of the outcome's support. For binary instruments the testable implications are sharp. Furthermore, I propose an implementation that links restrictions on latent response types to a generalization of first-order stochastic dominance and random utility models, allowing to distinguish violations of the exclusion restriction from violations of monotonicity-type assumptions. The testable implications extend naturally to the many instruments case. 
		
	\end{abstract}

	
	\clearpage
	
	\onehalfspacing 
	\section{Introduction}
	
	Consider a potential outcomes model defined by the following two equations:
	\begin{align}
		X_i=X_i(Z_i) &\textit{ a.s. }\\
		Y_i=Y_i(X_i,Z_i)&\textit{ a.s. }
	\end{align}
	Suppose that the joint distribution of $(Z_i,X_i,Y_i)$ over the support $\mathcal{Z}\times \mathcal{X}\times \mathcal{Y}$ is observed, and that both $\mathcal{X}$ and $\mathcal{Z}$ are finite while $\mathcal{Y}$ is unrestricted. The goal of this paper is to investigate the joint falsification of three assumptions: (i) that $Z$ takes values across the population as if it was randomly assigned, (ii) restrictions on the effect of $Z$ on $X$  (typically called instrument-response restrictions), and (iii) an exclusion restriction on $Z$, that is, the assumption that:
	\begin{align}
		Y_i=Y_i(X_i,Z_i)=Y_i(X_i)&\textit{ a.s. }
	\end{align}
	
	These three assumptions play an important role in modern econometrics. First, they are the cornerstone of a range instrumental variables (IV) methods for heterogeneous treatment effects in the spirit of \cite{angristImbens1994late}. Furthermore, as noted by \cite{kwon2024testing}, the three assumptions combined are equivalent to the sharp null of full mediation in the context of randomized control trials, that is, the hypothesis that a set of mediators $\mathcal{X}$ can fully account for the effect on $Y$ of an exogenous set of treatments $\mathcal{Z}$. 
	
	In this paper I propose a novel set of observable implications of the aforementioned assumptions. The observable implications are sharp when the instrument is binary, and extend naturally to the multiple instruments case. The core idea is to transform (via integration) a potentially infinite set of inequalities imposed by the exclusion restriction, into a finite set of linear restrictions on conditional treatment and response type probabilities. The resulting restrictions are characterized by the integral of the point-wise minimum of multiple conditional sub-densities.  
	
	In the many instruments case the observable implications have two advantages: (i) they use information from the join distribution of $(Z_i,X_i,Y_i)$ that previous approaches tailored to test IV validity do not incorporate, and (ii) they do not rely on monotonicity or support restrictions, allowing for general instrument-response restrictions. 
	
	Moreover, for  instrument-response restrictions that rule out certain behaviors or latent types, the characterization can be further simplified to a finite set of inequalities with a generalized first order stochastic dominance (FOSD) interpretation. This set of inequalities can be used to distinguish violations of the exclusion restriction from violations of instrument-response restrictions, and to link the latter to non-parametric assumptions about the potential preferences of units in the population. Thus, as a byproduct and in the absence of an exclusion restriction, these inequalities can also be used to falsify the assumption that some observed treatment choice probabilities are consistent with a set of pre-specified non-parametric restrictions on how an exogenous binary variable affects the preferences and choices of a population of rational decision makers. 
	
	To obtain the generalized FOSD inequalities I show that finding a distribution over potential treatments and outcomes that satisfies the sharp instrument validity conditions amounts to solving a well known discrete transportation problem. This technique, and its equivalence to the maximum flow/minimum cut problems, was first used to characterize incomplete econometric models by \cite{ekeland2010optimal} and \cite{galichon2011set}.\footnote{My approach differs in two ways: first, instead of considering transportation problems from the space of observed variables to the domain of latent unobservables, my transportation problem is defined between observed treatment probabilities conditional on different instrument values. Second, to incorporate the exclusion restriction I extend their minimum-cut characterization method to a broader class of capacity constraints.}
	
	My optimal transport approach is related to contemporaneous and independent work by \cite{kaido2025testing} who study the falsification of shape constraints and exclusion restrictions using graph theory. In general, neither of the two approaches is nested by the other. The main advantage of the approach from \cite{kaido2025testing} is its generality since it is sharp for the multiple instruments case and applies to other econometric models beyond IV. The main advantage of my approach is that it does not rely on partitioning the support of the outcome but instead  incorporates the sharp observable implications of the exclusion restriction as a capacity constraint on the transportation problem which yields a finite linear system of inequalities, even when the outcome has infinite support. 
	
	\medskip
	\textbf{The observable implications:} Let $\tilde{Z}\subset\mathcal{Z}$ be a subset of instrument values, and let $S(x_l,\tilde{Z})$ be the group%
	\footnote{The notation $S(x_l,\tilde{Z})$ comes from an acronym for \textit{sufficient takers}, i.e., units such that $Z_i\in\tilde{Z}$ is a sufficient condition for them to realize treatment $x_l$} %
	of units that would realize treatment $x_l$ whenever $Z_i\in\tilde{Z}$. Suppose that all units in $S(x_l,\tilde{Z})$ realize outcomes in some measurable set $B_Y\subset\mathcal{Y}$ when the treatment is set to $x_l$. Then, the units in $S(x_l,\tilde{Z})$ must be a subset of all the units that realize outcomes $Y_i\in B_y$ when $X_i=x_l$. Consequently, the mass of group $S(x_l,\tilde{Z})$ is bounded above by the minimum over $z_k\in\tilde{Z}$, of the conditional probability of observing $(Y_i\in B_Y,X=x_l)$ given $Z=z_k$. In other words, the mass of group $S(x_l,\tilde{Z})$ is bounded by the minimum over instrument values in $\tilde{Z}$ of the integrals of the sub-densities $\phi_{z_0,x_l}:=\mathbb{P}[X_i=x_l,Z=z_k]f_{y|Z_i=z_0,X_i=x_l}(y)$ over the set $B_Y$.

	A core result of this paper is to note that if the previous bound holds for every measurable subset of $\mathcal{Y}$, then, the instrument-wise minimum can be replaced by a point-wise minimum. That is, the probability that units in $S(x_l,\tilde{Z})$ realize outcomes in $B_Y$ is also bounded by the integral over $B_Y$ of the point-wise minimum $\psi(y)=\underset{z_k\in\tilde{Z}}{\min}\{\phi_{z_k,x_k}(y)\}$. Since this bound must hold for every measurable $B_Y$, setting $B_Y=\mathcal{Y}$ yields an upper bound on the probability that a unit belongs to $S(x_l,\tilde{Z})$. The result is a finite set of observable necessary conditions for instrument validity. 
	
	\medskip
	\textbf{Other related work:} This paper contributes to the literature on testing the validity of instrumental variables models. The first test for instrument validity for heterogeneous treatment effects was proposed by \cite{balke1997bounds} and subsequently refined by \cite{kitagawa2015test}, \cite{huber2015testing}, and \cite{mourifie2017testing}. These refinements, however, only apply when both the instrument and treatment are binary, and treatment adoption is weakly monotone in the instrument. \cite{kedagni2020generalized} proposed a joint test for the exclusion restriction and random assignment --without monotonicity or any other instrument-response restriction-- but their results are only sharp for a binary outcome. Subsequent work has derived sharp testable implications for IV designs on a case-by-case basis \citep{bai2024sharp}, or proposed general tests that are not sharp \citep{sun2023instrument}. More recently, \cite{kwon2024testing} provided the first sharp characterization for general binary instruments and \cite{kaido2025testing} for multi-valued instruments. 
	
	The observable implications I propose apply to a range of IV methods for heterogeneous treatment effects. \cite{angristImbens1994late} were the first to use non-parametric restrictions on instrument-response heterogeneity to identify causal parameters. Subsequent work has extended their approach to a wide range of settings. Major extensions include: re-interpretations of monotonicity (\cite{heckman2018unordered},  \cite{navjeevan2022ordered}); alternative restrictions on how the instrument affects treatment take-up (\cite{richardson2010analysis}, \cite{de2017tolerating}, \cite{dahl2023never}, \cite{goff2024does}); and generalizations to settings with multiple instruments and treatmens (\cite{lee2020treatment}, \cite{goff2020vector}, \cite{mogstad2021causal}, \cite{van2023limited}, \cite{kazemi2024instrumental}, \cite{fusejima2024identification}, \cite{bai2024inference}). These methodological advances have been applied in diverse empirical contexts, including education and labor market studies (\cite{kline2016evaluating}, \cite{kirkeboen2016field}, \cite{pinto2021beyond}, \cite{mountjoy2022community}). 
	
	The FOSD characterization I offer and its relation to non-parametric random utility models contributes to the literature that investigates the empirical content of response type restrictions (\cite{goff2024does} and \cite{bai2024identifying}) and their microeconomic foundations (\cite{vytlacil2002independence}, \cite{heckman2018unordered}, \cite{navjeevan2022ordered}).
	
	All the testable implications I provide can be represented by the requirement that a linear system of inequalities admits at least one solution. Inference on the feasibility of linear systems and linear programs is an active area of research (\cite{bai2022testing}, \cite{fang2023inference}, \cite{cho2024simple}, \cite{goff2025inference}) tightly related (and in some cases equivalent) to inference methods for models characterized by moment inequalities such as \cite{cox2023simple} or \cite{andrews2023inference}. 
	
	The remainder of this paper is organized as follows: In section \ref{Section:Setup} I introduce the class of models and restrictions considered. To illustrate the intuition behind the observable implications and compare my approach to the existing literature, in section \ref{Section:Binary} I consider the binary instrument case. In section \ref{Subsection:OTandFOSD} I specialize the characterization from section \ref{Section:Binary} to instrument-response-type restrictions and link them to FOSD and random utility models using optimal transport. In section \ref{Section:ManyIV} I extend the observable implications to the multiple instruments case. Section \ref{Section:Conclusion} concludes
	
	\section{Instrument validity: Notation and framework}\label{Section:Setup}
	
	\subsection{The IV model with discrete instruments and treatments}
	
	The population consists of units indexed by $i\in\mathcal{I}$. I assume that the joint distribution of the outcome $Y_i$, treatment $X_i$, and instrument $Z_i$ is observed, that is, the joint distribution of $(Y,X,Z)$ is identified. I impose no restrictions on the cardinality or dimension of the support of the outcome, which I denote as $\mathcal{Y}$.  The treatment and instrumental variables have finite supports $\mathcal{X}=\{x_0,x_1,...,x_{L-1}\}$ and $\mathcal{Z}=\{z_0,z_1,...,z_{K-1}\}$ respectively. 
	
	To formalize the structure of $\mathcal{Y}$, I assume it is equipped with a $\sigma-$algebra $\mathcal{B}_Y$ and a measure $\mu_Y$ so that $(\mathcal{Y},\mathcal{B}_Y,\mu_Y)$ forms a measure space. Analogous constructions hold for $\mathcal{X}$ and $\mathcal{Z}$, which, being finite, are naturally endowed with power set $\sigma$-algebras and counting measures. 
	
	I assume that all other relevant covariates are exogenous. Since they play no role in the characterizations I provide, I assume that all the identification arguments are conditional on such covariates and omit them from the formal proofs and results. 
	
	Unless otherwise noted, all the results will be presented and discussed using the IV notation and terminology. All results would also hold for the sharp null of full mediation after appropriately relabeling the values of $X$ as mediators and the values of $Z$ as treatments.
	
	I adopt the potential outcomes framework. $X_i(z_k)$ denotes the treatment that unit $i$ would realize if assigned to instrument value $z_{k}$, and $Y_i(x_{l},z_{k})$ denotes the outcome under treatment $x_l$ and instrument $z_k$. The observed values of $Y_i$ and $X_i$ correspond almost surely to the potential outcomes and treatments evaluated at the realized values of $Z_i$ and $X_i$. 
	\begin{align}
		X_i=&X_i(Z_i) \textit{ a.s.}\\
		Y_i=&Y_i(X_i,Z_i)=Y_i(X_i(Z_i),Z_i) \textit{ a.s.}
	\end{align}	 

	For $Z$ to be a valid instrument, two assumptions are typically invoked, instrument independence and exclusion:
	
	\begin{assumption}[Instrument validity]\label{Assumption:Validity}\phantom{a}
		\begin{enumerate}
			\item \textbf{Instrument independence}:
			\begin{align}
				\mathbb{P}[Y_i(x_{l'},z_{k'})\in B_Y,X_i(z_{k''})=x_{l}|Z_i=z_k]=\mathbb{P}[Y_i(x_{l'},z_{k'})\in B_Y,X_i(z_{k''})=x_{l}] .
			\end{align}
			For every $\phantom{a}x_l,x_{l'}\in\mathcal{X},\phantom{a}z_k,z_{k'},z_{k''}\in\mathcal{Z}$ and $B_Y\in\mathcal{B}_Y$.
			\item \textbf{Exclusion restriction}:
			\begin{align}
				Y_i(x_{l},z_k)=Y_i(x_{l}) \textit{ a.s. }
			\end{align}	
			For every $\phantom{a}x_l\in\mathcal{X}$ and $ z_k\in\mathcal{Z}$
		\end{enumerate} 
	\end{assumption}
	
	Put in words, the instrument must be independent from potential outcomes and treatments as if it was randomly assigned, and it must have no effect on outcomes after accounting for treatment. 

	Assumption \ref{Assumption:Validity} implies that the observed joint distribution of $(X,Y,Z)$ is fully determined by the observed distribution of $Z$, and the unobserved distribution of potential outcomes and potential treatments. This motivates the following definitions:
	
	\begin{definition}\phantom{a}
		\begin{enumerate}
			\item An instrument-response type $t^Z\in\mathcal{T}_Z$ is a function $t^Z:\mathcal{Z}\longrightarrow \mathcal{X}$ that maps instrument values to treatment values.
			\item A treatment-response type $t^X\in\mathcal{T}_X$ is a function $t^X:\mathcal{X}\longrightarrow \mathcal{Y}$ that maps treatment values to outcome values.
			\item A response type $t$ is a pair $t=(t_z,t_x)\in\mathcal{T}=\mathcal{T}_Z\times\mathcal{T}_X$.
		\end{enumerate}
	\end{definition}
	
	Thus, units are fully characterized by their response type $t\in\mathcal{T}$. That is, their potential treatment and outcome realizations over the support of $X$ and $Z$.  An important feature of this definition of response types is that by construction they satisfy the exclusion restriction. The reason is that for any $x_l\in\mathcal{X}$, a response type $t=(t^Z,t^X)$ assigns a unique outcome value $t^X(x_l)$ when $X=x_l$ regardless of the value that the instrument takes.

	Note that since $\mathcal{X}$ is finite, any instrument-response type $t^Z$ can be represented as a $|\mathcal{Z}|$-dimensional vector $(t^Z(z_1),...,t^Z(z_k),...,t^Z(z_{K}))$. Similarly, any treatment-response type $t^X$ can be represented as a $|\mathcal{X}|$-dimensional vector $\big(t^X(x_1),...,t^X(x_{L})\big)$. Therefore, in a slight abuse of notation, I will refer to both instrument-response and treatment-response types in their vector form so that $t^Z=\big(t^Z(z_1),...,t^Z(z_{K})\big)$ and $t^X=\big(t^X(x_1),...,t^X(x_l),...,t^X(x_{L})\big)$. Consistent with this notation, for $1\leq k \leq K$ and $1\leq l\leq L$, I denote the $k-th$ element of $t^Z$ and the $l-th$ element of $t^X$ as $(t^Z)_k$ and $(t^X)_{l}$ respectively. In other words, $(t^Z)_k$ denotes the treatment value that instrument-response type $t^Z$ realizes when $Z=z_k$, and $(t^X)_{l}$ denotes the outcome that treatment-response type $t^X$ realizes when $X=x_{l}$. 
	
	Furthermore, any response type $t=(t^Z,t^X)$ can be represented as the concatenation of the vector representations of $t^Z$ and $t^X$. This representation results in a $|\mathcal{Z}|+|\mathcal{X}|$ dimensional vector that belongs to the set  $\mathcal{T}:=\mathcal{T}_Z\times\mathcal{T}_X=\mathcal{X}^{|\mathcal{Z}|}\times\mathcal{Y}^{|\mathcal{X}|}$.
	
	$\mathcal{T}$ is the Cartesian product of a finite number of finite sets, each identical to $\mathcal{X}$, and a finite number of sets identical to $\mathcal{Y}$. Hence, the product measure $\mu_T:=\mu_X\otimes\mu_X\otimes...\mu_X\otimes\mu_Y\otimes\mu_Y\otimes...\otimes\mu_Y$ with respect to the product $\sigma$-algebra $\mathcal{B}_T:=\mathcal{B}_X\otimes\mathcal{B}_X\otimes...\mathcal{B}_X\otimes\mathcal{B}_Y\otimes\mathcal{B}_Y\otimes...\otimes\mathcal{B}_Y$ over $\mathcal{T}$ is well defined. Having equipped $\mathcal{T}$ with a measure I can formally define a distribution over response and instrument-response types.	
	
	\begin{definition}\phantom{a}
		\begin{enumerate}
			\item A distribution over response types is a probability measure $g:\mathcal{B}_T\longrightarrow [0,1]$ over the measure space $(\mathcal{T},\mathcal{B}_T)$.
			\item A distribution over instrument-response types is a function $p:\mathcal{T}_Z\longrightarrow[0,1]$ such that $p(t^Z)\geq 0$ for all $t^Z\in\mathcal{T}_Z$, and $\sum\limits_{t^Z\in\mathcal{T}_Z}p(t^Z)=1$.
		\end{enumerate}
	\end{definition}
	
	\begin{remark}\phantom{a}\label{Remark:gImpliesf}Any distribution over response types $g$ defines a unique distribution over instrument-response types $f_g$ according to the following equation:
		\begin{align}
			\phantom{aaaaaaaaaa}p_g(t'^Z):=&g\Big(\big\{t=(t^Z,t^X)\in\mathcal{T}:t^Z=t'^Z,t^X\in\mathcal{Y
			}^{L}\big\}\Big), &\textit{ }\forall\phantom{a} t'^Z\in\mathcal{T}_Z,\\
			\phantom{aaaaaaaaaa}=&g\Big(\big\{t=(t^Z,t^X)\in\mathcal{T}:t^Z=t'^Z\big\}\Big), &\textit{ }\forall \phantom{a}t'^Z\in\mathcal{T}_Z.
		\end{align}	
	\end{remark}
	
	If assumption \ref{Assumption:Validity} holds, then an unobserved true distribution over response types $g_0$ must exist that generates the observed distribution of $(Y,X,Z)$. 
	\begin{definition}[Consistency]\label{Definition:Consistency}\phantom{a}
		A distribution over response types $g\in\Delta(\mathcal{T})$ is consistent with observed conditional outcome and treatment probabilities if:
		\begin{align}	\label{Equation:Consistency}
			\mathbb{P}[y_i\in B_Y,X_i=x_{l}|Z_i=z_k]=g_0\Big(\big\{(t^Z,t^X)\in\mathcal{T}:(t^Z)_k=x_l,(t^X)_l\in B_Y\big\}\Big),\\
			\forall x_l\in\mathcal{X}, z_k\in\mathcal{Z}, \textit{ and } B_Y\in\mathcal{B}_Y \nonumber.
		\end{align}
	\end{definition}
	
	In the previous equation, the observed conditional probability on the left hand side must correspond to the mass that $g_0$ assigns to response types whose potential treatment is equal to $x_l$ when the instrument takes the value $z_k$ and their potential outcome is contained in $B_Y$ when treatment is equal to $x_l$. 
	
	Let $\Delta(\mathcal{T})$ be the set of all probability measures over $(\mathcal{T},\mathcal{B}_T)$ (i.e. distributions over response types), assumption \ref{Assumption:Validity} is falsified if no $g\in\Delta(\mathcal{T})$ exists that satisfies equation $\ref{Equation:Consistency}$. 
	
	\subsection{Restrictions on instrument-response heterogeneity}
	
	To facilitate the identification of causal parameters using the IV method, assumptions that restrict instrument-response heterogeneity are often invoked. The most prominent example is monotonicity in the binary instrument and treatment case which, as \cite{angristImbens1994late} showed, allows identification of a local average treatment effect. Subsequent work has explored both relaxations and extensions of this assumption to settings with multiple instruments and multiple, ordered and unordered, treatments (\cite{heckman2018unordered},\cite{goff2020vector}, \cite{mogstad2021causal}, \cite{van2023limited}, \cite{lee2020treatment}, \cite{bai2024sharp}), these generalizations typically rely on ruling out particular instrument-response types.
	
	Alternatively, researchers may impose inequality restrictions on the prevalence or relative frequency of certain instrument-response types—for example, requiring that the proportion of compliers exceeds that of defiers as \cite{de2017tolerating}, or bounding the share of defiers for sensitivity analysis as in \cite{huber2014sensitivity}, \cite{noack2021sensitivity} or \cite{yap2025sensitivity}). The framework of this paper accomodates all the aforementioned restrictions on their own and combined.
	
	To introduce the class of instrument-response restrictions I consider, note that the total number of possible instrument-response types is $|\mathcal{T}_Z|:=|\mathcal{X}|^{|\mathcal{Z}|}=(L)^{K}$. Therefore, a distribution $p$ over $\mathcal{T}_Z$ can be represented as a vector $\boldsymbol{p}$ in $\mathbb{R}_+^{|\mathcal{T}_Z|}$, where each coordinate corresponds to the probability that $p\in\Delta(\mathcal{T}_Z)$ assigns to a different $t\in\mathcal{T}_Z$. 
	
	To formally define $\boldsymbol{p}$ is suffices to note that the set of response types is finite and thus can be well ordered. Then, the $n-th$ coordinate of $\boldsymbol{p}$ represents the probability that the distribution $p\in\Delta(\mathcal{T}_Z)$ assigns to the $n-th$ response type. The specific ordering of response types is not relevant as long as it is well defined.
	
	The class of restrictions on heterogeneity I consider take the form $A_Z\boldsymbol{p}\leq \boldsymbol{r}$ where $A_Z$ is an $M$ by $|\mathcal{T}_Z|$ matrix with known coefficients for some fixed $M\geq 1$, and $\boldsymbol{r}$ is an $M$ dimensional vector with entries that are either known or point identified from the data. I use the subscript $Z$ to emphasize that the matrix $A_Z$ encodes restrictions that pertain to treatment assignment as a function of the instrument.
	
	\begin{assumption}\phantom{a}\label{Assumption:SelectionModel} 
		Let $M\geq$ be a fixed integer, $A_Z$ a known fixed matrix $A_Z\in\mathbb{R}^{M\times|\mathcal{T}_Z|}$, and $\boldsymbol{r}_Z\in\mathbb{R}^{M}$ a point identified vector that depends on the joint distribution of $(Y,X,Z)$.
		
		A distribution over response types $g_0\in\Delta(\mathcal{T})$ exists such that if $\boldsymbol{p}_{g_0}$ denotes the vector representation of $p_{g_0}$, the following system of inequalities holds:
		\begin{align}\label{Equation:SelectionModel}
			A_Z\boldsymbol{p}_{g_0}\leq \boldsymbol{r}_Z.
		\end{align}	
	\end{assumption} 
	Note that strict equality restrictions can be incorporated by representing them as pairs of inequality constraints. In particular, if the $k$-th entry of $\boldsymbol{p}$ corresponds to the instrument-response type $t^Z\in\mathcal{T}_Z$, the restriction that $t^{Z}$ occurs with zero probability can be imposed by including the canonical basis vectors $e_k$ and $-e_k$ as rows of $A_Z$, and setting the corresponding entries of $\boldsymbol{r}_Z$ to zero.%
	\footnote{$e_k$ denotes the vector that takes the value 0 in every entry except the $k$-th which is equal to 1.} %
	Thus, assumption \ref{Assumption:SelectionModel} can accommodate restrictions that rule-out specific response types. In appendix \ref{Appendix:RestrictionExamples} I provide examples of models and restrictions encompassed by assumption \ref{Assumption:SelectionModel}. 
	
	To conclude, note that if no distribution over instrument-response types $p\in\Delta(\mathcal{T}_Z)$ exists such that $\boldsymbol{p}$ satisfies equation \ref{Equation:SelectionModel} and is consistent with observed conditional treatment probabilities, then, the restrictions on instrument-response heterogeneity imposed by assumption \ref{Assumption:SelectionModel} are falsified. 
	
	\begin{proposition}[\textbf{Sharp observable implications of assumption \ref{Assumption:SelectionModel}}]\label{Proposition:SharpSelectionModel_LinearEquations}\phantom{a}A distribution over instrument-response types $p$ is consistent with assumption \ref{Assumption:SelectionModel} and observed conditional treatment probabilitites $\mathbb{P}[X_i=x_l|Z_i=x_l]$, if and only if it satisfies the following system of linear equalities and inequalities:	
		\begin{align}
			A_{Z}\boldsymbol{p}&\leq \boldsymbol{r}_Z\label{Equation:SM1},\\
			A_{X}\boldsymbol{p}&=\boldsymbol{r}_X\label{Equation:SM2}.
		\end{align} 
		Where $A_Z$ and $\boldsymbol{r}_Z$ are defined as in assumption \ref{Assumption:SelectionModel} and the system $A_X\boldsymbol{p}=\boldsymbol{r}_X$ stacks the following set of equations:
		\begin{align}
			\phantom{aaaaa}\mathbb{P}[X_i=x_{l}|Z_i=z_k]&=\sum\limits_{t^Z\in\mathcal{T}_Z:t^Z(z_k)=x_l}p(t^Z), \phantom{aaaaa}\forall z_{l}\in\mathcal{Z},x_k\in\mathcal{X}.
		\end{align}	
	\end{proposition}
	The proof of proposition \ref{Proposition:SharpSelectionModel_LinearEquations} is straightforward and follows from the previous discussion. 
	
	\section{The binary instrument case} \label{Section:Binary}
	
	Assumption \ref{Assumption:IVResponseTypeRestriction} and the observed conditional treatment probabilities restrict the set of admissible instrument-response probabilitites. Now I show that, for binary instruments, the restrictions implied by the observed conditional distribution of the outcome combined with the exclusion restriction can be expressed as a finite set of linear restrictions on instrument-response probabilities too. This insight reduces a potentially infinite set of inequalities to a finite one. 
	
	To this end, I first introduce some additional notation. Define the function $\Psi_{x_l}:\mathcal{Y}\longrightarrow \mathbb{R}_{+}$ as follows:
	\begin{align}
		\psi_{x_l}(y)=\underset{k=0,1}{min}\left\{\boldsymbol{P}_{x_l|z_k}\left[f_{Y|X=x_{l},Z=z_{k}}(y)\right]\right\}.
	\end{align}
	Where $\boldsymbol{P}_{x_l|z_k}:=\mathbb{P}[X=x_l|Z=z_k]$ and $f_{Y|X=x_{l},Z=z_{k}}$ denote the distribution of $Y$ conditional on $X=X_k$ and $Z=Z_k$.
	
	Now define the quantity:
	\begin{align}
		\Psi_{x_l}=\int\limits_{\mathcal{Y}}\psi_{x_l}(y)d\mu_Y.
	\end{align}
	
	Finally, for any $x_l\in \mathcal{X}$, define the set of always-$x_l$ takers as units such that $X_i(z_0)=X_i(z_1)=x_l$ and note that a distribution over response types $p\in\Delta(\mathcal{T}_Z)$ assigns probability $p\left(AT_{x_l}\right)$ to always-$x_l$ takers where $p\left(AT_{x_l}\right)$ is defined as follows:
	\begin{align}
		p\left(AT_{x_l}\right):=\sum\limits_{t^Z\in\mathcal{T}_Z}p(t^Z)\boldsymbol{I}_{\{t^Z(z_0)=t^Z(z_1)=x_l\}}.
	\end{align}
	Where $\boldsymbol{I}_{\{t(z_0)=t(z_1)\}}$ is an indicator that response type $t^Z$ realizes treatment $x_l$ at both values of $Z$.
	
	Finally, Let $t^Z_j$ denote the $j$-th instrument-response type according to the fixed pre-specified order over instrument-response types. Let $A_{Y}$ be a $L\times|\mathcal{T}_Z|$ matrix with entires $(A_{Y})_{l,j}=\boldsymbol{I}_{\{t^Z_j(z_0)=t^Z_j(z_1)=x_l\}}$ and let $\boldsymbol{\Psi}$ denote the vector that stacks the parameters $\Psi_{x_0},...,\Psi_{x_K}$ and has $\Psi_{x_l}$ in its $l$-th coordinate. 
	
	The following theorem characterizes the sharp observable implications of binary instruments for heterogeneous treatment effects:
	\begin{theorem}\label{Theorem:SharpBinary}[Sharp observable implications of binary instruments]\phantom{a}\\
		The following statements are equivalent:
		\begin{enumerate}
			\item A distribution over response types $g\in\Delta(\mathcal{T})$ that satisfies, assumptions \ref{Assumption:Validity} (instrument validity), and \ref{Assumption:SelectionModel} (instrument-response restriction) and definition \ref{Definition:Consistency} (consistency with observed probabilities) exists. 
			\item A distribution over instrument-response types $p\in\Delta(\mathcal{T}_Z)$ exists such that $p$ satisfies:
			\begin{align}
				A_{Z}\boldsymbol{p}&\leq \boldsymbol{r}_Z,\\
				A_{X}\boldsymbol{p}&=\boldsymbol{r}_X,\\
				A_{Y}\boldsymbol{p}&\leq \boldsymbol{\Psi}.
			\end{align}
		\end{enumerate}
		Where $A_Z,A_X,r_Z$ and $r_X$ are finite dimensional matrices as defined in assumption \ref{Assumption:SelectionModel} and proposition \ref{Proposition:SharpSelectionModel_LinearEquations}.
	\end{theorem}
	The formal proof of theorem \ref{Theorem:SharpBinary} is given in appendix \ref{Appendix:ProofBinaryIV}. Here, I provide some intuition on the result and compare it to existing falsification procedures.
	
	To see why the inequality $p(AT_{x_l})\leq \Psi_{x_l}$ must hold for every $x_l\in\mathcal{X}$ note that always-$x_l$ takers are a potentially strict subset of all the units that could potentially realize treatment $x_{l}$ both when $Z=z_0$ and when $Z=z_1$. Therefore, if the exclusion restriction holds, the group of always-$x_l$ takers that realize outcomes in any measurable set $B_Y\subset\mathcal{B}_Y$ must also be a (potentially strict) subset of the units that realize outcomes in $B_Y$ when $X=x_l$ and $Z=Z_k$ for $z_k\in\{z_0,z_1\}$. This observation is illustrated in figure \ref{Figure:Necessity}.
	
	To see the sufficiency of the restrictions, suppose that some $p\in\mathcal{T}_Z$ is consistent with observed conditional treatment probabilities and satisfies $p(AT_{x_l})\leq\Psi_{x_l}$ for every $x_\in\mathcal{X}$. To construct a distribution over response types consistent with the observed joint distribution of $(Y,X,Z)$, the unobserved marginal distribution of potential outcomes for always-$x_l$ takers when $X=x_l$ can be set to $\big[\frac{p(AT_{x_l})}{\Psi_{x_l}}\big]\psi_{x_l}(y)$ so that always-$x_l$ takers account for a fraction $\frac{p(AT_{x_l})}{\Psi_{x_l}}$ of the mass represented by $\Psi_{x_l}$. This construction guarantees that the unconditional sub-density associated to the outcomes of always-$x_l$-takers is dominated by $\boldsymbol{P}_{x_l|z_0}\left[f_{Y|X=x_{l},Z=z_{0}}(y)\right]$ and $\boldsymbol{P}_{x_l|z_1}\left[f_{Y|X=x_{l},Z=z_{1}}(y)\right]$. 

	\FloatBarrier
	\begin{figure}[h!]
		\caption{Upper bounds on always takers.}\label{Figure:Necessity}
		\begin{minipage}{0.48\textwidth}
			\centering
			\includegraphics[width=0.99\linewidth]{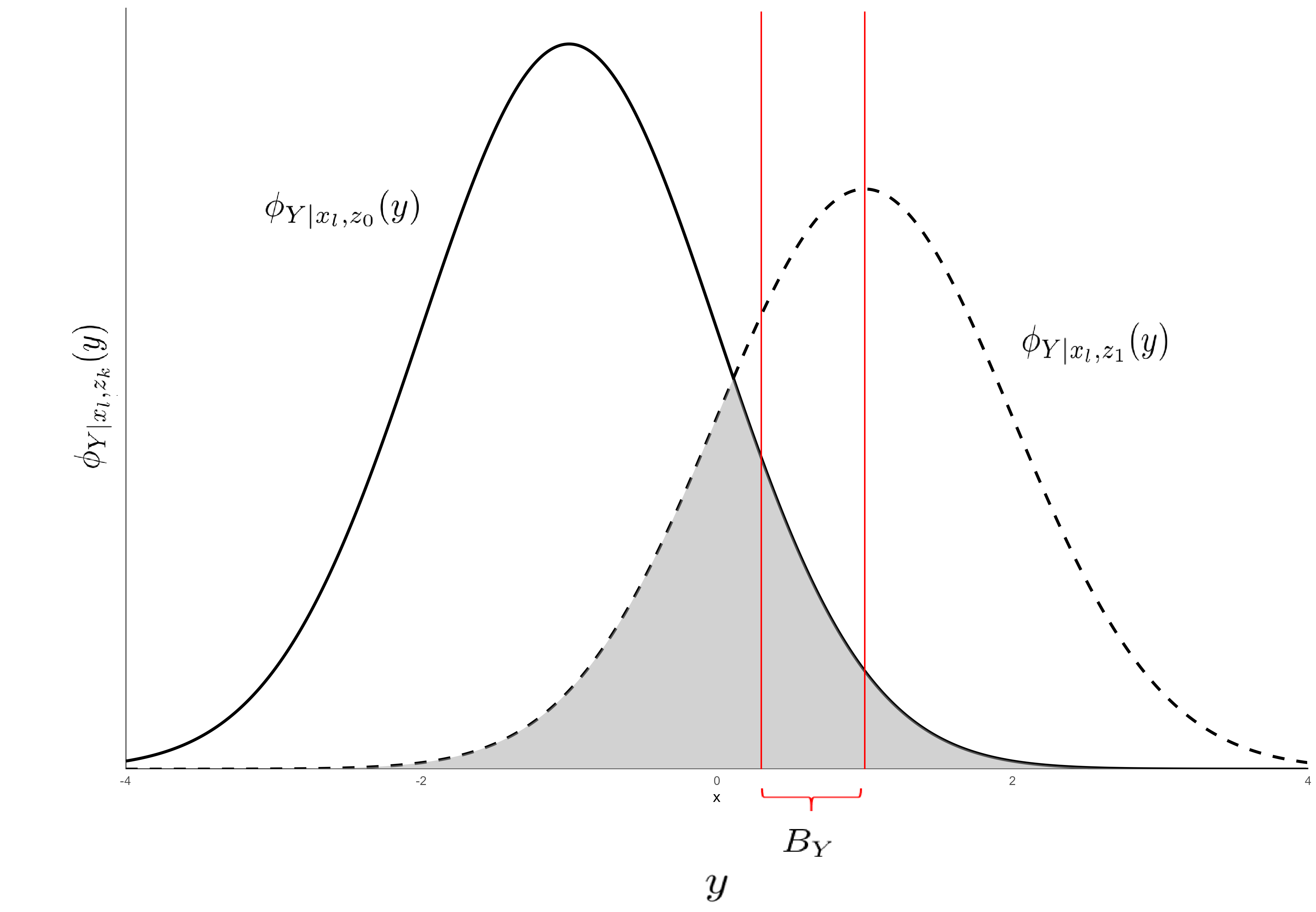}
			\subcaption{}\label{Figure:Necessity_a}
		\end{minipage}
		\begin{minipage}{0.48\textwidth}
			\centering
			\includegraphics[width=0.99\linewidth]{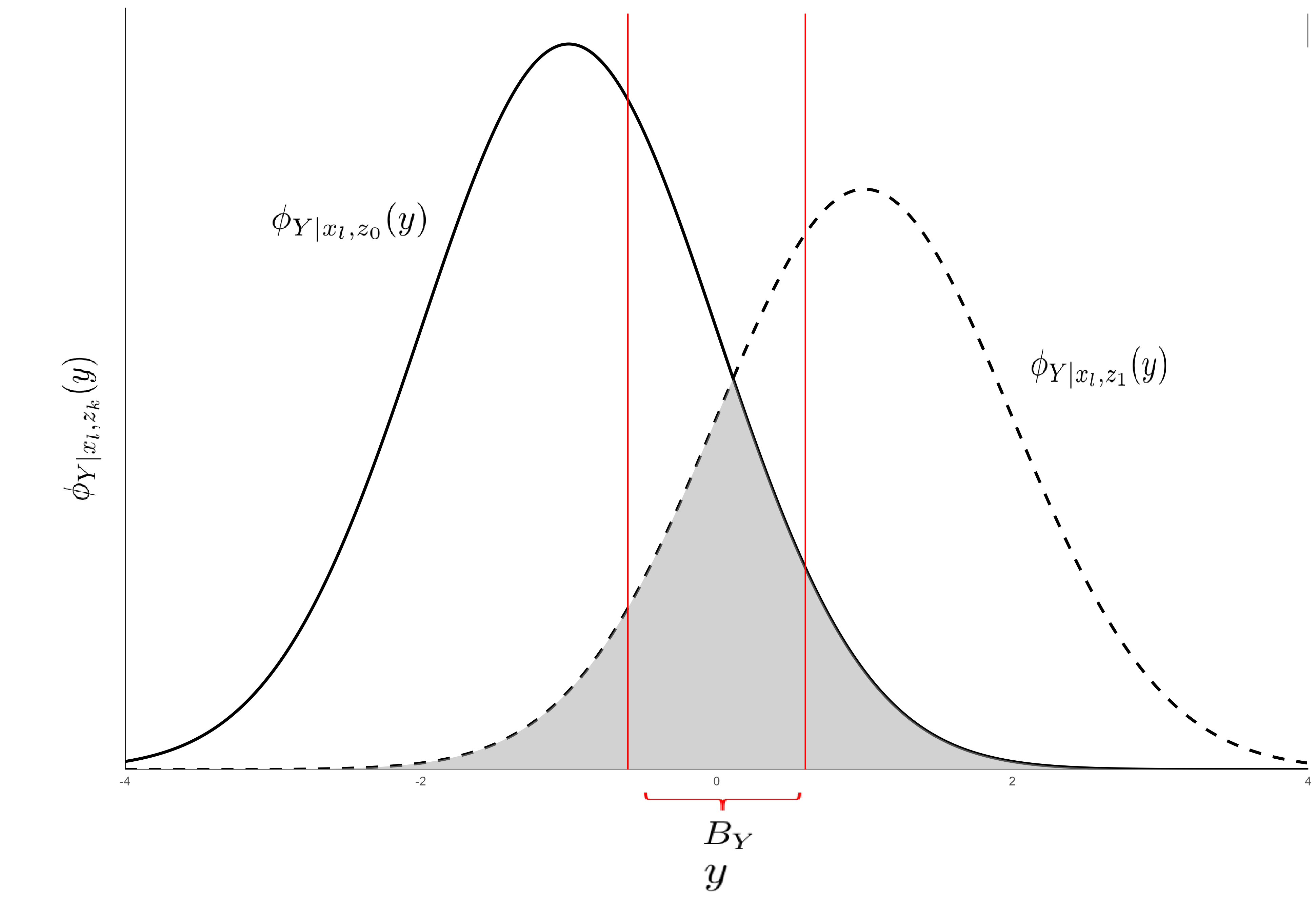}
			\subcaption{}\label{Figure:Necessity_b}
		\end{minipage}
		\footnotesize \textit{\textbf{Notes:} In both sub-figures the solid and dashed black lines respectively represent the hypothetical sub-densities $\phi_{Y|x_l,z_0}=\boldsymbol{P}_{x_l|z_0}\left[f_{Y|X=x_{l},Z=z_{0}}(y)\right]$ and $\phi_{Y|x_l,z_1}=\boldsymbol{P}_{x_l|z_1}\left[f_{Y|X=x_{l},Z=z_{1}}(y)\right]$. The constant $\Psi_{x_l}$ is defined as the integral over $\mathcal{Y}$ of the point-wise minimum of these two sub-densities, thus, its magnitude corresponds to the gray colored area. The two vertical red lines delimit a measurable interval $B_Y\in\mathcal{B}_Y$.}
		
		\textit{The exclusion restriction implies that the distribution of outcomes for the group $AT_{x_l}$ of subjects that realize treatment $x_l$ at both instrument values ($z_0$ and $z_1$) must be the same conditional on $Z_i=z_0$ or $Z_i=z_1$. Moreover, the units in $AT_{x_l}$ whose outcomes belong to some measurable $B_Y\in\mathcal{B}_Y$ are a (potentially strict) subset of all the units that would realize treatment $x_l$ and outcomes in $B_Y$. As a consequence, conditional on $Z_i=z_k$, the probability that the outcomes of such group belong to any measurable set $B_Y\in\mathcal{B}_Y$ must be lower that the probability that the events $\{Y_i\in B_Y\}$ and $\{X_i=x_l\}$ are jointly realized. For instance, in figure \ref{Figure:Necessity_a}, the fraction of units that realize treatment $x_l$ and outcomes in $B_Y$ conditional on $Z=z_0$ strictly coincides with the intersection of the gray area and the rectangle delimited by the red lines. On the other hand, conditional on $Z=z_1$, there is an additional group of units (whose mass corresponds to the area delimited by the red lines, the dashed black line, and the contour of the gray area) who realize outcomes in $B_Y$. This additional group cannot belong to $AT_{x_l}$, otherwise, their outcomes would change without their treatment changing, which would be a violation of the exclusion restriction. 
			In figure \ref{Figure:Necessity_b} none of the two sub-densities delimits the upper contour of the gray area, and therefore, the integral of the point-wise minimum yields a tighter bound  on always takers that realize outcomes in $B_Y$ than the minimum of the integrals of $\phi_{Y|x_l,z_0}$ and $\phi_{Y|x_l,z_1}$ over $B_Y$. This argument is formalized in appendix \ref{Appendix:Proofs}.}
	\end{figure}

	It remains to distribute the fraction $1-\frac{p(AT_{x_l})}{\Psi_{x_l}}$ of the area $\Psi_{x_l}$ as well as any mass strictly above at least one sub-density across response types that realize $x_l$ at just one instrument value (i.e. $x_l$-compliers and $x_l$-defiers). Assigning potential outcomes to these units is simple because --since they realize treatment $x_l$ at just one instrument value-- the exclusion restriction does not bind their potential outcomes. Figure \ref{Figure:Sufficiency_a} provides a visual illustration of this procedure. The formal construction is described in appendix \ref{Appendix:Proofs}. 
	
	\begin{figure}[h!]
		\caption{Sufficiency of the necessary conditions}
		\begin{minipage}{0.48\textwidth}
			\centering
			\includegraphics[width=0.99\textwidth]{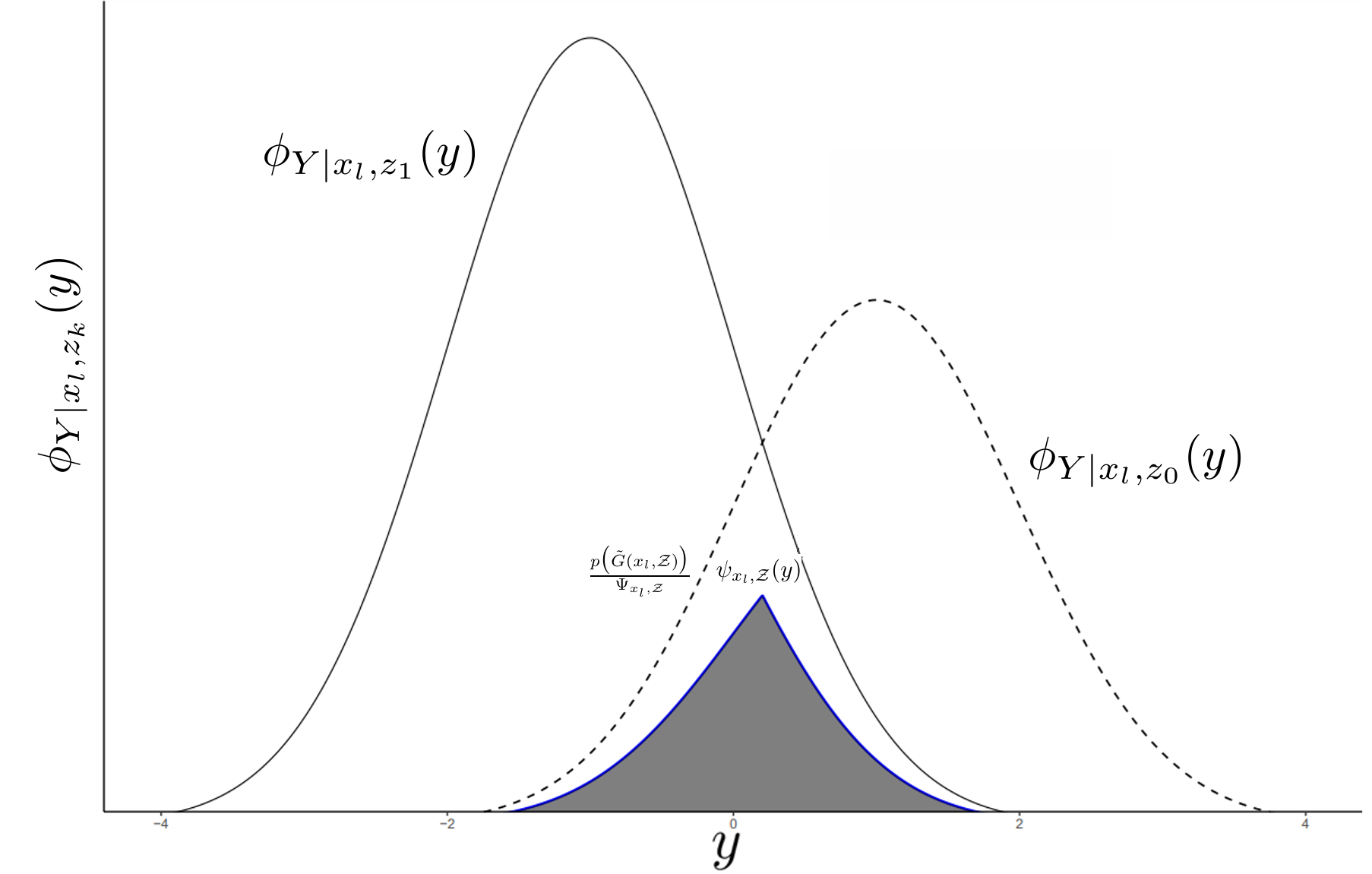}
			\subcaption{General case}\label{Figure:Sufficiency_a}
		\end{minipage}
		\begin{minipage}{0.48\textwidth}
			\centering
			\includegraphics[width=0.99\textwidth]{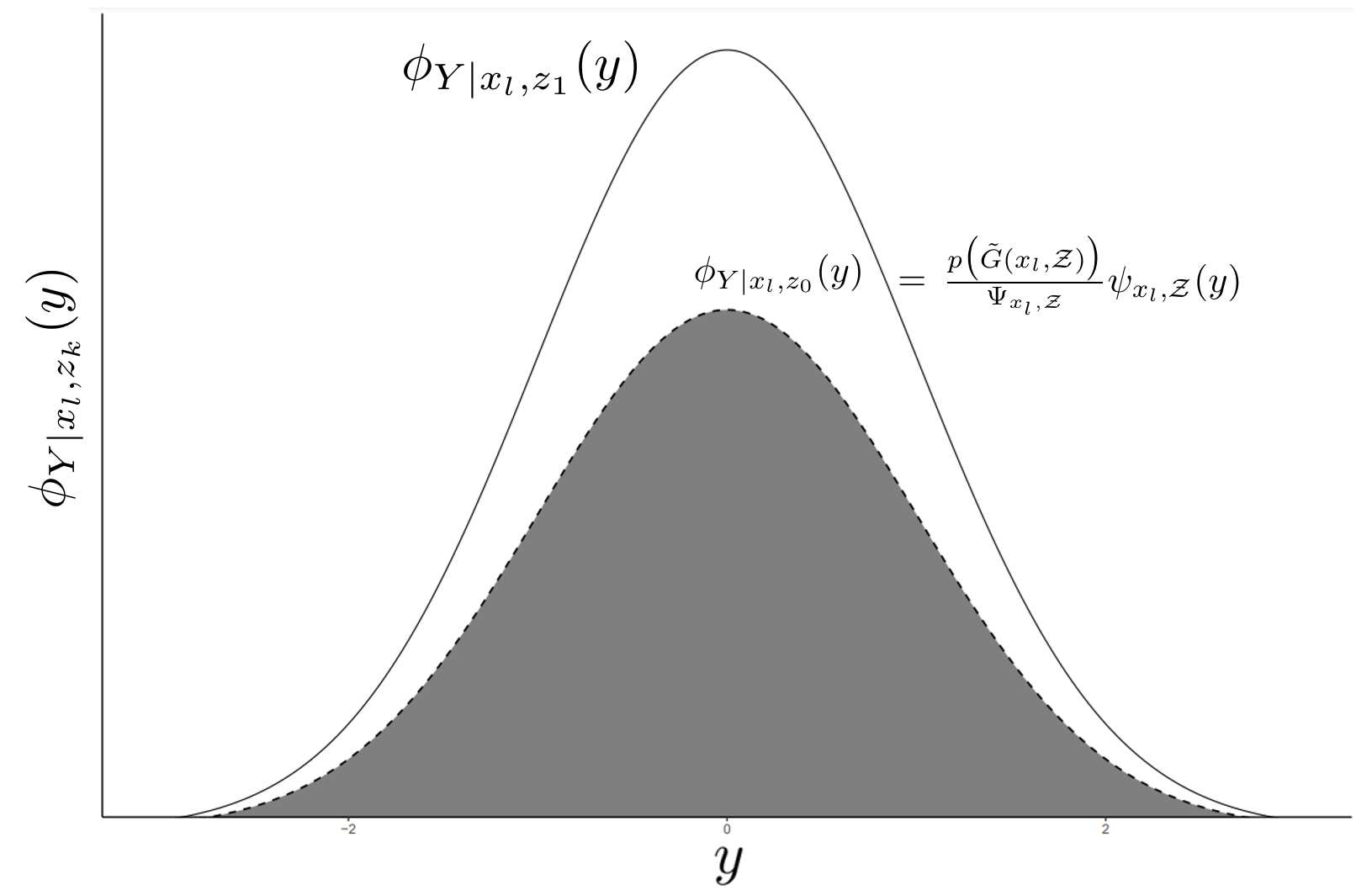}
			\subcaption{Monotone IV}\label{Figure:Sufficiency_b}
		\end{minipage}
		\footnotesize \textit{\textbf{Notes:} In both sub-figures the solid and dashed black lines respectively represent the hypothetical sub-densities $\phi_{Y|x_l,z_0}:=\boldsymbol{P}_{x_l|z_0}\left[f_{Y|X=x_{l},Z=z_{0}}(y)\right]$ and $\phi_{Y|x_l,z_1}:=\boldsymbol{P}_{x_l|z_1}\left[f_{Y|X=x_{l},Z=z_{1}}(y)\right]$. The function $\psi_{x_l}(y)$ is defined as the point-wise minimum $\psi_{x_l}(y):=min\{\phi_{Y|x_l,z_0}(y),\phi_{Y|x_l,z_1}(y)\}$, and the constant $\Psi_{x_l}$ as its integral over $\mathcal{Y}$, thus, its magnitude corresponds to the area that lies strictly below both sub-densities. In this figure, the gray area equals the quantity $p(AT_{x_l})$. The blue solid line represents the sub-density $\frac{p(AT_{x_l})}{\Psi_{x_l}}\psi_{x_l}(y)$ which is proportional to the distribution of potential outcomes for always-$x_l$-takers. Then, the distribution of potential outcomes for $x_l$-compliers and $x_l$-defiers can be defined to be proportional to the sub-density $\boldsymbol{P}_{x_l|z_k}\left[f_{Y|X=x_{l},Z=z_{k}}(y)\right]-\frac{p(AT_{x_l})}{\Psi_{x_l}}\psi_{x_l}(y)$ for $k=0,1$ respectively. }
		
		\textit{Sub-figure (a) corresponds to the general case without monotonicity. It also illustrates a situation in which a valid distribution over instrument response types $p\in\Delta(\mathcal{T}_Z)$ is such that $p(AT_{x_l})<\Psi_{x_l}$. The monotone case is illustrated in sub-figure (b). In this case, monotonicity enforces that $p(AT_{x_l})=\Psi_{x_l}$, furthermore, the distribution of potential outcomes for always-$x_l$-takers is identified and is proportional to $\psi_{x_l}(y)$.}
	\end{figure}	
	\normalsize
	\textbf{Comparison with \cite{kitagawa2015test} and \cite{kwon2024testing}}: A benchmark case is when assumption \ref{Assumption:SelectionModel} imposes that treatment $x_l$ can only increase (or decrease) as the instrument changes. This case corresponds to the classical monotone binary instrument and treatment model from \cite{angristImbens1994late},\cite{kitagawa2015test} and \cite{mourifie2017testing}. It also corresponds to maximal elements in the ordered or unordered case as in \cite{sun2023instrument}. 
	
	Suppose that addoption of $x_l$ weakly increases as $Z$ changes from $z_0$ to $z_1$. Note that in this special case, the function $\boldsymbol{P}_{x_l|z_1}\left[f_{Y|X=x_{l},Z=z_{1}}(y)\right]$ weakly dominates the function $\boldsymbol{P}_{x_l|z_0}\left[f_{Y|X=x_{l},Z=z_{0}}(y)\right]$, therefore, $\psi_{x_l}(y)$ is equal to $\boldsymbol{P}_{x_l|z_0}\left[f_{Y|X=x_{l},Z=z_{0}}(y)\right]$ almost everywhere. Furthermore, the fraction of always $x_l$ takers is identified and coincides with $\Psi_{x_l}$ and $\boldsymbol{P}_{x_l|z_0}$. This particular case is illustrated in figure \ref{Figure:Sufficiency_b}.
	
	Figure \ref{Figure:Sufficiency_b} illustrates why, while seemingly simpler, my approach is equivalent to \cite{kitagawa2015test} in the binary monotone case. The integral of $\psi_{x_l}(y)$ over $\mathcal{Y}$ constrains the frequency of always takers, this statement is true even if monotonicity does not hold. The test of \cite{kitagawa2015test} takes the function $\boldsymbol{P}_{x_l|z_0}\left[f_{Y|X=x_{l},Z=z_{0}}(y)\right]$ and asks if, indeed, it coincides with $\psi_{x_l}$ almost everywhere. Instead, I first compute the integral of $\psi_{x_l,\mathcal{Z}}$ without invoking monotonicity. Then, I note that if violations of monotonicity occur with positive probability, $\psi_{x_l,\mathcal{Z}}$ will be strictly smaller than $\boldsymbol{P}_{x_l|z_0}\left[f_{Y|X=x_{l},Z=z_{0}}(y)\right]$ in some positive measure set. Thus, the two integrals are equal if and only if random assignment, exclusion, and monotonicity are not violated.
	
	For the non-monotone case, \cite{kwon2024testing} note that the population frequency of $x_l$-defiers (units that realize $X_i=x_l$ only when $Z_i=z_0$) is bounded below by the supremum over measurable sets $B_Y\in\mathcal{B}_Y$ of the difference $\mathbb{P}[Y_i\in B_Y,X_i=x_l|Z_i=z_0]-\mathbb{P}[Y_i\in B_Y,X_i=x_l|Z_i=z_0]$. Note that this supremum is attained in the set $\left\{y\in\mathcal{Y}: \boldsymbol{P}_{x_l|z_0}\left[f_{Y|X=x_{l},Z=z_{0}}(y)\right]\geq \boldsymbol{P}_{x_l|z_1}\left[f_{Y|X=x_{l},Z=z_{1}}(y)\right]\right\}$ and its value coincides with the quantity $\mathbb{P}[X_i=x_l,Z_i=z_0]-\Psi_{x_l}$. Since at $Z=z_0$ only $x_l$-defiers and $x_l$-always takers can realize $x_l$, the upper bound on always takers I derive implies and is implied by the lower bound on compliers from \cite{kwon2024testing}. The difference is that instead of taking a supremum over a potentially infinite collection of measurable sets, I obtain an equivalent bound by integrating the point-wise minimum of two sub-densities. 
	
	\section{Generalized FOSD and response type restrictions}\label{Subsection:OTandFOSD}	
	
	\subsection{Sharp observable implications of response type restrictions for binary instruments}
	
	In this section I specialize theorem \ref{Theorem:SharpBinary} to instrument-response-type restrictions, i.e. restrictions that rule out from the population specific instrument-response types. The most common example of a restriction in this class is the no-defiers assumption from \cite{angristImbens1994late}. The specialized result is a set of sharp inequalities that only depend on point identified parameters.
	
	The resulting inequalities can be interpreted as a generalization of FOSD. This interpretation allows me to link response type restrictions to non-parametric random utility models. Specifically, I show that any instrument-response type restriction can be derived from the assumption that, as the instrument changes from $z_0$ to $z_1$, it affects certain pairwise comparisons in one direction and therefore prevents specific preference reversals over treatments.
	
	The refined inequalities come from representing the instrument validity conditions as a flow problem. The observed conditional treatment probability $\mathbb{P}[X_i=x_l|Z_i=z_0]$ represents the share of the population that realizes treatment $x_l$ at instrument value $z_0$. In the absence of instrument-response type restrictions, these units could realize any other treatment at the instrument value $z_1$. But instrument-response type restrictions rule out certain potential treatments for units that realize treatment $x_l$ when $Z_i=z_k$. That is, flows from $x_k$ to certain treatment values $\tilde{X}_{x_k}\subset \mathcal{X}$ are not allowed.%
	\footnote{This restriction is reminiscent of the minimal monotonicity idea discussed by \cite{navjeevan2022ordered}, but it need not apply to every pair of treatment values, just to a subset.}%
	
	The refined observable implications come from taking the flow idea literally. Finding a valid distribution over response types is equivalent to finding a valid set of flows that, for every $x_l\in\mathcal{X}$, takes the share $\mathbb{P}[X_i=x_l|Z_i=z_0]$ of the population and distributes it across potential counterfactual treatments $X_i(z_1)$ in a way that matches the observed treatment probabilities conditional on $Z_i=z_1$ without recurring to the flows that the instrument-response type restriction forbids. 
	
	To formally define the transportation problem I require additional notation and definitions. I begin by formally defining response type restrictions.
	
	\begin{assumption}[Instrument-response type restriction]\label{Assumption:IVResponseTypeRestriction}\phantom{a}\\
		Let $R\subset\mathcal{T}_Z=\mathcal{X}\times \mathcal{X}$ be a fixed and known collection of instrument-response types. A distribution over response types $g_0\in\Delta(\mathcal{T})$ exists such that $p_{g_0}(t^Z)=0$ for all $t^Z\in R$.
	\end{assumption}
	Put in words, a response type restriction simply imposes that some instrument-response types occur with zero probability.
	
	Now define the sets $\mathcal{X}_k$ for $k=0,1$ as copies of $\mathcal{X}$ but with elements $x_l^k\in\mathcal{X}_k$ labeled with the superscript $k$ to make explicit that $x^0_l$ refers to treatment $x_l\in\mathcal{X}$ realized when $Z=0$ and $x^1_l$ to treatment $x_l\in\mathcal{X}$ realized when $Z=1$.
	
	Next, define the directed graph with vertex set $\mathcal{X}_0\cup \mathcal{X}_0$ and let the edges be  $\mathcal{E}_R\subset\mathcal{X}_{1}\times\mathcal{X}_{2}$, where the ordered pair $(x_l^0,x_{l'}^1)\in\mathcal{X}_0\times\mathcal{X}_1$ belongs to $\mathcal{E}_R$ if and only if the instrument-response type $t^Z=(x_l,x_{l'})$ does not belong to $R$. In order words, an edge $(x_l^0,x_{l'}^1)$ belongs to the graph if the response type that realizes $x_l$ at $Z=0$ and $x_{l'}$ at $Z=1$ is not ruled out by assumption \ref{Assumption:IVResponseTypeRestriction}. 
	
	For instance, in the case of the no defiers assumption, the set $R$ would consist of the singleton $\{(x_1,x_0)\}$. Therefore, the edges $(x_0^0,x_0^1), (x_0^0,x_1^1)$ and $(x_1^0,x_1^1)$ --corresponding to never takers, compliers, and always takers-- would belong to $\mathcal{E}_R$, in contrast, the edge $(x_1^0,x_0^1 )$ would not. This particular case is illustrated in figure \ref{Figure:GraphRep}.
	
	\FloatBarrier
	\begin{figure}[h!]
		\caption{Graph representation of the no-defiers instrument-response type restriction.}\label{Figure:GraphRep}
		\centering
		\includegraphics[width=0.45\textwidth]{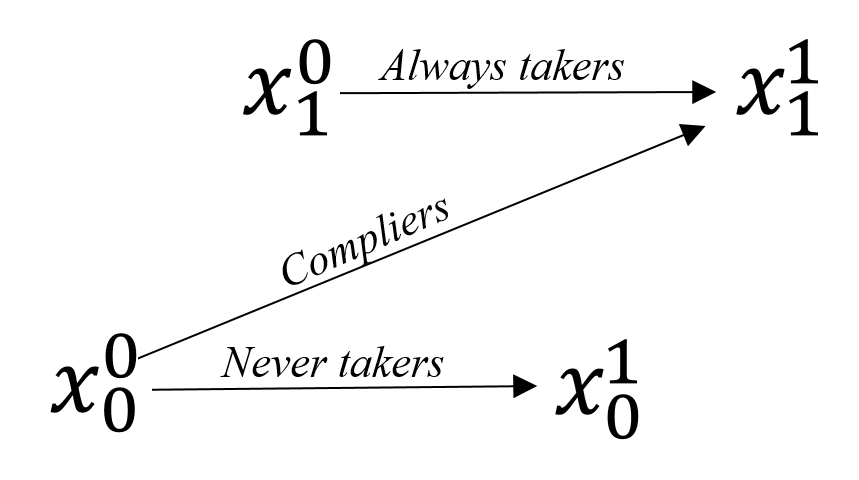}\\
		\medskip
		\justifying\footnotesize\textit{\textbf{Notes:} This figure depicts the graph representation of the no-defiers/monotonicity assumption from \cite{angristImbens1994late}. The edges with a super-index $k=0,1$ represent treatment values realized at instrument values $Z=z_k$. The edges with sub-index $l=0,1$ represent treatments $x_0$ and $x_1$ respectively. Since there is no edge from $x_1^0$ to $x_0^1$, defiers are ruled out from the model. }
	\end{figure}
	\normalsize
	\FloatBarrier

	The transportation problem is as follows: every node $x_l^0\in\mathcal{X}_0$ starts with a mass $\mathbb{P}[X_i=x_l|Z_i=z_0]$ of units. Each unit can travel from its origin node $x_l^0$ through the admissible edges $(x_l^0,x_{l'}^1)\in\mathcal{E}_R$ in order to reach some destination node $x_{l'}^1\in\mathcal{X}_1$ subject to the constraint that each destination node $x_{l'}$ cannot receive more than $\mathbb{P}[X_i=x_{l'}|Z_i=z_1]$ units. If it is possible to empty all origin nodes $x_l^0$ so that all destination nodes end at full capacity, then, the transportation problem is solved. 
	
	If the transportation problem is solved, a distribution over response types exists that is consistent with both the observed joint distribution of $(X,Y)$ and the  imposed response type restriction. In appendix \ref{Appendix:Proofs} I provide the formal definition of the transportation problem and prove the equivalence between the existence of a feasible solution and a valid distribution over response types. The central argument, however, is simple: There is a one to one mapping between response types in $R$ and edges in $\mathcal{E}_R$. From here, it is easy to see that the value of the flow that travels through edge $e$ according to a solution to the transportation problem corresponds to the probability that a valid distribution over instrument-response types assigns to the unique instrument-response type associated to $e$.
	
	Two well known combinatorial problems are equivalent to the transportation problem described: the maximum flow and minimum cut problems.%
	\footnote{The equivalence is also proven in appendix \ref{Appendix:ProofFOSD}. The entire reasoning would also apply to more general restrictions such as those considered in assumption \ref{Assumption:SelectionModel}. The only difference would be that more general restrictions cannot necessarily be stated as capacity constraints on the edges of the graph. This impossibility would complicate using the equivalence with minimum-cut problem to obtain observable inequalities that do not depend on unobservables.} The minimum cut formulation is particularly useful to derive the sharp observable implications of incomplete models in terms of inequalities that only depend on observable quantities as first noted by \cite{ekeland2010optimal} and \cite{galichon2011set} due to its equivalence to plausibility constraints in the context of random set theory (\cite{beresteanu2012partial}). 
	
	Before stating the observable implications, I make explicit the one to one correspondence between response type restrictions and binary relations that exists in the binary instrument case, and introduce a generalization of the lower contour to arbitrary binary relations.
	
	\begin{definition}[Binary relation associated to $R\subset\mathcal{T}_Z$]\phantom{a}\\
		Let $R\subset\mathcal{T}_Z$ be a collection of response types. Define the binary relation $\geq^R\subset\mathcal{X}\times\mathcal{X}$ as follows:
		\begin{align}
			x_l\geq^R x_{l'}\Longleftrightarrow (x_l,x_{l'})\in R
		\end{align}
	\end{definition}
	That is, $x_l$ is related to $x_{l'}$ according to $\geq^{R}$ if switches from $x_l$ to $x_{l'}$ when the instrument goes from $z_0$ to $z_1$ are ruled out by the instrument-response type restriction associated to the set $R$.
	
	\begin{definition}[Common lower contour of $\geq^R$]\phantom{a}\\
		Let $\geq^R$ be a binary relation over $\mathcal{X}$. For any $S\subset\mathcal{X}$, the common lower contour of $S$ according to $\geq^R$ is:
		\begin{align}
			L_{\geq^R}(S)=\{x_k\in\mathcal{X}:s \geq^R x_{k},\forall s\in S\}
		\end{align}
	\end{definition}
	In the context of instrument-response types, $L_{\geq^R}(S)$ is the (potentially empty) subset of elements of $\mathcal{X}$ such that no flows from an element in $S$ to any element in $L_{\geq^R}(S)$ are allowed. Thus, its complement $L_{\geq^R}(S)^C$ can be interpreted as the set of treatments that units that start with treatment values in $S$ are weakly incentivized to adopt as the instrument changes from $z_0$ to $z_1$. Note that for total orders (or their linear extensions), the set $L_{\geq^R}(S)$ coincides with the standard definition of the upper contour of $S$ and the set  $L_{\geq^R}(S)^C$ with the upper contour.
	
	Now I state the sharp observable implications of response type restrictions for binary instruments: 
	
	\begin{theorem}\label{Theorem:FOSD}\phantom{a}\\
		A distribution over instrument-response types $p\in\Delta(\mathcal{T}_Z)$ that satisfies assumption \ref{Assumption:Validity} (instrument validity), assumption \ref{Assumption:IVResponseTypeRestriction} (instrument-response type restriction), and definition \ref{Definition:Consistency} (consistency) exists, if and only if: 
		\begin{enumerate}
			\item The following inequality holds for every $S\subset\mathcal{X}$.
			\begin{align}
				\mathbb{P}[X_i\in S |Z_i=z_0]\leq \mathbb{P}[X_i\in L_{\geq^{R}}(S)^C|Z_i=z_1]\label{Equation:FOSD}
			\end{align}
			\item For any $S\subset\mathcal{X}$ and any partition $(\Lambda_S,\Lambda'_S)$ of $L_{\geq^R}(S)^C$ satisfying (i) $\emptyset\neq\Lambda'_S\subset S$, (ii) $L_{\geq^R}(x_{l'})^C\setminus\{x_{l'}\}\subset \Lambda_S$ for every $x_{l'}\in \Lambda_S'$, and (iii) $L_{\geq^R}(S\setminus \Lambda_S')^C\subset \Lambda_S$:
			\begin{align}
				\mathbb{P}[X_i\in S |Z=z_0]\leq \mathbb{P}[X_i\in \Lambda_S|Z=z_1]+\sum\limits_{x_{l'}\in \Lambda'_S}  \Psi_{x_{l'}} \label{Equation:GeneralizedFOSD}
			\end{align}
			Where $\Psi_{x_l}$ is defined as in the previous section. 	
		\end{enumerate}
	\end{theorem}
	
	The proof of theorem \ref{Theorem:FOSD} is given in appendix \ref{Appendix:ProofFOSD}. One interpretation of the inequalities from part one is that mutually exclusive events cannot occur in the population with probability greater than one, this condition is implicitly verified by the minimum cut problem. A generalization of a logically equivalent idea beyond the binary instrument case was recently developed by \cite{kaido2025testing}.
	
	Another way to interpret the inequalities from part one  of theorem \ref{Theorem:FOSD} is generalized FOSD. The response type restriction associated to $R$ implicitly assumes that, if $x_l\geq^R x_{l'}$, then treatment $x_l$ is incentivized by the instrument relative to $x_{l'}$. Thus, the relative likelihood of observing $x_{l}$ relative to $x_{l'}$ must increase. As such, this set of inequalities can be independently used as a sharp test for the validity of response type restrictions or, as formalized in the next sub-section, to test non parametric hypotheses about how an exogenous variable affects the preferences of a population of decision makers. 
	
	The second part of theorem \ref{Theorem:FOSD} notes that the exclusion restriction and the observed conditional distribution of the outcome bind the prevalence of always-$x_l$ takers, and, in some cases, these constraints on always takers, \textit{tighten} the FOSD inequalities by further limiting the maximum of units that could potentially realize treatments in certain upper contours.
	
	The fact that the exclusion restriction \textit{tightens} the generalized FOSD inequalities by replacing some conditional treatment probabilities $\mathbb{P}[X_i=x_{l'}|Z=z_1]$ with the weakly smaller quantities $\Psi_{x_{l'}}$ allows to distinguish violations from the exclusion restriction from violations of instrument-response restrictions. In particular, one of the following three cases must occur:
	\begin{enumerate}
		\item All inequalities are satisfied. Hence, neither the exclusion restriction nor the response-type restriction is falsified.
		\item Only inequalities from part 2 are violated. Hence, the response type restriction is not falsified on its own, but only when combined with the exclusion restriction.  
		\item At least one inequality from part 1 is violated. Hence, the response type restriction is falsified
	\end{enumerate}
	Additionally, every inequality is linked to exactly one set $S$. Thus, if the inequalities associated with $S$ are not satisfied, it must be that the assumed restriction on how the instrument affects the potential treatments (or preferences over treatments) of units that realize treatments in $S$ when $Z_i=z_0$ is falsified, while restrictions associated to other treatment switches might not.
	
	One final advantage of theorem \ref{Theorem:FOSD} relative to theorem \ref{Theorem:SharpBinary} is that it yields a linear system of inequalities with no nuisance parameters. 
	
	To illustrate its usefulness, I use theorem \ref{Theorem:FOSD} to characterize the binary monotone instrument IV model from \cite{angrist1995two}.
	
	\begin{corollary}[\textbf{Sharp observable implications of binary monotone instruments}]\label{Corollary:BinaryBinary}
		
		Suppose that $\mathcal{X}=\{x_0,x_1,...,x_{L}\}$, with $x_{l}<x_{l'}$ for every $l<l'$. A distribution over response types exists that satisfies assumptions \ref{Assumption:Validity} and:
		\begin{align}
			X_i(z_0)\leq X_i(z_1) \phantom{aaa} a.s.\label{Equation:Monotonicity}
		\end{align}
		If and only if:
		\begin{align}
			\mathbb{P}[X\geq x_l|Z=z_0]\leq \Psi_{x_l} + \mathbb{P}[X> x_{l+1}|Z=z_1] \phantom{aaa}\forall x_l\in\mathcal{X} \label{Equation:ExclusionFOSD}
		\end{align}
		
	\end{corollary} 
	
	Note that in the ordered monotone case, the set $L_{\geq}(x_l)^{C}$ corresponds to the standard upper contour of $x_l$ (i.e. $\{x_{l},x_{l+1},...,x_L\}$), thus, the sharp inequalities reduce to FOSD but \textit{tightened} by the exclusion restriction in the sense that the probability $\mathbb{P}[X= x_{l}|Z=z_1]$ associated to the minimal element of every $L_{\geq}(x_l)^{C}$ (i.e. $x_l$) is replaced by the weakly smaller quantity $\Psi_{x_l}$.%
	\footnote{As a consequence of this result, it is easy to see that the conditions from \cite{sun2023instrument} do not incorporate the effect that the exclusion restriction has on upper contours associated to non minimal/maximal elements. This illustrates why  the observable implications from \cite{sun2023instrument} are not sharp for monotone instruments, and why they are nested by the ones proposed in this paper.}

	\subsection{A submonotone interpretation for response type restrictions}\label{Subsection:Submonotonicity}
	
	If treatment is assumed to be chosen by rational decision makers, the connection between binary relations, response type restrictions, and first order stochastic dominance extends beyond realized treatments to utility maximization. This relation does not depend on functional or distributional assumptions --such as additive separability. Instead, it can be directly expressed in terms of latent potential preferences. 
	
	\begin{assumption}[$\geq^{*}$-Submonotonicity]\label{Assumption:SM}
		Let $\geq^{*}$ be an irreflexive%
		\footnote{The extension to reflexive binary relations is conceptually simple. The reasoning remains the same but requires special treatment of response type restrictions that rule out ``always takers'' of one particular treatment. In the interest of brevity this case is addressed in appendix \ref{Appendix:ProofSubmonotonicity}.} %
		binary relation over $\mathcal{X}$. Suppose that:
		\begin{enumerate}
			\item Units are endowed with strict preference profiles over $\mathcal{X}$ that depend on $Z_i$ denoted $\succ_i^{(z_k)}$ for $k=0,1$ where the function $\succ_i^{(z_k)}$ maps $\mathcal{Z}$ to strict preference profiles over $\mathcal{X}$. In other words, potential preferences.
			\item Units are rational in the sense that $X_i(z_k)=x_l$ implies $x_l \succ_i^{(z)} x_{l'}$ almost surely  for all $l,l'\in\{0,1,...,L\}$ such that $l\neq l'$.
			\item There are no preference reversals against $\geq^*$ when the instrument changes from $z_0$ to $z_1$. That is, if $x_l\phantom{a}\geq^*x_{l'}$ and $x_l\succ_i^{(z_0)} x_{l'}$, then, $x_l\succ_i^{(z_1)} x_{l'}$ almost surely.
		\end{enumerate}
		
	\end{assumption}
	
	The first two points of assumption \ref{Assumption:SM} simply posit that response types arise from latent utility maximization. The third point captures the the notion that if $x_l\geq^*x_{l'}$, then alternative $x_l$ is being weakly incentivized or promoted (not necesarilly relative to all other alternatives but relative to $x_{l'}$) by the instrument.
	
	I call assumption \ref{Assumption:SM} submonotonicity because it encodes a notion weaker than monotonicity; it rules out two-way flows between certain pairs of treatments, but not necessarily between all. In this sense, any submonotonicity assumption can always be extended to a full monotonicity condition with respect to a total order (or the linear extension of a total order) by incorporating additional restrictions. For instance, both ordered and unordered monotonicity correspond to special cases of submonotonicity where the binary relation is transitive. The notion of preventing ``two-way flows'' from \cite{navjeevan2022ordered} corresponds to assuming that the binary relation in question is asymmetric. 
	
	In Appendix \ref{Appendix:Submonotonicity} I explore in depth the interpretation of different assumptions within the submonotonicity class and provide multiple motivating examples. These examples include extensions of unordered monotonicity (\cite{heckman2018unordered}); assumptions that draw from both ordered and unordered monotonicity (\cite{rose2024recoding}); specific cases of targeting and encouragement designs as well as some possible relaxations and generalizations (\cite{lee2020treatment}, \cite{bai2022testing}); and designs in which the instrument encourages variety and experimentation, or its effect depends on \textit{status-quo} ($Z=z_0$) choices.
	
	\begin{theorem}\label{Theorem:Submonotonicity}
		Let $\geq^{*}$ be an irreflexive binary relation over $\mathcal{X}$. Then:
		\begin{enumerate}
			\item If $\geq^{*}$-Submonotonicity holds, then, $p_{g_0}(t^Z)=0$ for every $t^Z\not\in R$. 
			\item Suppose that $R\subset\mathcal{T}$ is such that $\geq^{R}=\geq^*$. Then, $R$ is the smallest cardinality subset of $\mathcal{T}_Z$ such that if $p_{g_0}(t)=0$ for all $t\not\in R$, then, $\geq^{*}$-Submonotonicity cannot be falsified.
		\end{enumerate}
	\end{theorem}
	
	The assumption that $\geq^R$ is irreflexive is not necessary, but it substantially simplifies exposition. The proof of theorem \ref{Theorem:Submonotonicity} is straightforward and can be found in appendix \ref{Appendix:ProofSubmonotonicity} where the general case is covered.
	
	Theorem \ref{Theorem:Submonotonicity} shows that $\geq^R$-submonotonicity is sufficient to guarantee that instrument-response types outside of $R$ occur with zero probability. Moreover, $R$ is the minimal (in the sense of set inclusion) restriction over instrument-response types, such that $\geq^R$-submonotonicity is not falsified. Thus, the set of inequalities from  \ref{Equation:FOSD} comprises the sharp observable implications of $\geq^{R}$-submonotonicity.
	
	\section{The multiple instruments case}\label{Section:ManyIV}
	
	\subsection{Sufficient takers: A finite set of necessary conditions}
	
	The bounds on always takers from the previous sections can be naturally extended to the multiple instruments case. In this more general setting, the exclusion restriction imposes upper bounds on the prevalence of multiple groups of instrument-response types in addition to always takers. The following definition formalizes such instrument-response type groups.
	
	\begin{definition}[Sufficient-$x_l$-takers]
		For any treatment value $x_l\in\mathcal{X}$ and any non-empty subset of instrument values $\tilde{Z}\subset\{z_0,z_1,...,z_{K_{Z}-1}\}$ define the set of response types that realize $x_l$ whenever $Z$ takes any value in $\tilde{Z}$ as:
		\begin{align}
			\tilde{S}(x_l,\tilde{Z}):=\{t^Z:(t^Z)_k=x_l \phantom{a} \Longrightarrow z_k\in\tilde{Z}, \}
		\end{align} 
	\end{definition}
	
	For any subset $\tilde{Z}\subset\mathcal{Z}$ with $|\tilde{Z}|\geq 2$, define the function $\psi_{x_l\tilde{Z}}(y)$ and the constant $\Psi_{x_l,\tilde{Z}}$ as follows:
	\begin{align}
		\psi_{x_l,\tilde{Z}}(y):=&\underset{\tilde{k}\in\tilde{Z}}{min}\big\{\boldsymbol{P}_{x_l|z_k}\left[f_{Y|X=x_{l},Z=z_{k}}(y)\right]\big\}\\
		\Psi_{x_l,\tilde{Z}}:=&\int\limits_{\mathcal{Y}}\psi_{x_l,\tilde{Z}}(y)d\mu_Y
	\end{align}
	As in the binary instrument case, the function \( \psi_{x_l, \tilde{Z}}(y) \) represents the point-wise minimum of the sub-distributions \( \phi_{Y \mid x_l, z_{\tilde{k}}} \) over all \( z_k \in \tilde{Z} \). The constant \( \Psi_{x_l, \tilde{Z}} \) measures the area that lies below all these subdensities (the overlap). 
	
	\begin{proposition}\label{Proposition:UpperBoundOnFT}Suppose assumption \ref{Assumption:Validity} holds, then:
		\begin{align}
			\mathbb{P}[Y_i\in B_Y,X_i(z_k)=x_l\phantom{a}\forall z_k\in\tilde{Z}] \leq \int\limits_{B_Y}\psi_{x_l,\tilde{Z}}(y)d\mu_Y \phantom{aaa} \forall B_Y\in\mathcal{B}_Y
		\end{align}
	\end{proposition}
	Therefore, for any collection of instrument values $\tilde{Z'}$, setting $B_Y=\mathcal{Y}$ yields an upper bound on the probability that a unit belongs to the sufficient-$x_l$-takers group $\tilde{S}(x_l,\tilde{Z})$. In terms of $\boldsymbol{p}\in\Delta(\mathcal{T}_Z)$:
	\begin{align}
		\sum\limits_{t^Z\tilde{S}(x_l,\tilde{Z})}p(t^Z)\leq\Psi_{x_l,\tilde{Z}}\phantom{a}\forall\tilde{Z}\in\mathcal{P}(\mathcal{Z})\setminus\{\emptyset\}, \textit{ and } x_l\in\mathcal{X}
	\end{align}
	
	Once more, the previous equation comprises a finite  system of ineequalities  that are linear in $\boldsymbol{p}$. Therefore, it can be expressed in matrix form as $A_\Psi\boldsymbol{p}\leq \boldsymbol{\Psi}$ where $\boldsymbol{\Psi}$ stacks all the constants $\Psi_{x_l,\tilde{Z}}$ across treatment values and subsets of the support of the instrument. 
	
	\begin{corollary}\label{Corollary:UpperBoundOnFT}
		A valid distribution over instrument response types must satisfy:
		\begin{align}
			A_\Psi\boldsymbol{p}&\leq \boldsymbol{\Psi}\\
			A_{Z}\boldsymbol{p}&\leq \boldsymbol{r}_Z\\
			A_{X}\boldsymbol{p}&=\boldsymbol{r}_X
		\end{align}
	\end{corollary}
	
	\section{Conclusion} \label{Section:Conclusion}
	This paper proposes a set of observable implications for IV models. For binary instruments, the observable implications are sharp. For binary instruments with instrument-response-type restrictions I provide a characterization with a generalized FOSD interpretation and use this interpretation to link instrument-reponse-type restrictions to non-parametric assumptions on the preferences of units in the population. The testable implications generalize naturally to the multiple instruments case.
	
	\clearpage
	
	\singlespacing
	\footnotesize
	\bibliography{references.bib}	

@article{fang2023inference,
	title={Inference for Large-Scale Linear Systems With Known Coefficients},
	author={Fang, Zheng and Santos, Andres and Shaikh, Azeem M and Torgovitsky, Alexander},
	journal={Econometrica},
	volume={91},
	number={1},
	pages={299--327},
	year={2023},
	publisher={Wiley Online Library}
}

@techreport{bai2022testing,
	title={On testing systems of linear inequalities with known coefficients},
	author={Bai, Yuehao and Santos, Andres and Shaikh, Azeem M},
	year={2022},
	institution={Working Paper}
}

@article{cho2024simple,
	title={Simple inference on functionals of set-identified parameters defined by linear moments},
	author={Cho, JoonHwan and Russell, Thomas M},
	journal={Journal of Business \& Economic Statistics},
	volume={42},
	number={2},
	pages={563--578},
	year={2024},
	publisher={Taylor \& Francis}
}

@article{cox2023simple,
	title={Simple adaptive size-exact testing for full-vector and subvector inference in moment inequality models},
	author={Cox, Gregory and Shi, Xiaoxia},
	journal={The Review of Economic Studies},
	volume={90},
	number={1},
	pages={201--228},
	year={2023},
	publisher={Oxford University Press}
}

@article{andrews2023inference,
	title={Inference for linear conditional moment inequalities},
	author={Andrews, Isaiah and Roth, Jonathan and Pakes, Ariel},
	journal={Review of Economic Studies},
	volume={90},
	number={6},
	pages={2763--2791},
	year={2023},
	publisher={Oxford University Press US}
}

@article{goff2025inference,
	title={Inference on the value of linear programs},
	author={Goff, Leonard and Mbakop, Eric},
	journal={arXiv preprint arXiv:2506.06776},
	year={2025}
}

@article{angristImbens1994late,
	author = {Guido W. Imbens and Joshua D. Angrist},
	journal = {Econometrica},
	number = {2},
	pages = {467--475},
	publisher = {[Wiley, Econometric Society]},
	title = {Identification and Estimation of Local Average Treatment Effects},
	volume = {62},
	year = {1994}
}

@article{angrist1995two,
	title={Two-stage least squares estimation of average causal effects in models with variable treatment intensity},
	author={Angrist, Joshua D and Imbens, Guido W},
	journal={Journal of the American statistical Association},
	volume={90},
	number={430},
	pages={431--442},
	year={1995},
	publisher={Taylor \& Francis}
}

@article{heckman2018unordered,
	title={Unordered monotonicity},
	author={Heckman, James J and Pinto, Rodrigo},
	journal={Econometrica},
	volume={86},
	number={1},
	pages={1--35},
	year={2018},
	publisher={Wiley Online Library}
}

@techreport{navjeevan2022ordered,
	title={Ordered, Unordered and Minimal Monotonicity Criteria},
	author={Navjeevan, Manu and Pinto, Rodrigo},
	year={2022},
	institution={UCLA, mimeo, https://www. rodrigopinto. net}
}

@article{de2017tolerating,
	title={Tolerating defiance? Local average treatment effects without monotonicity},
	author={De Chaisemartin, Clement},
	journal={Quantitative Economics},
	volume={8},
	number={2},
	pages={367--396},
	year={2017},
	publisher={Wiley Online Library}
}

@article{dahl2023never,
	title={It is never too LATE: a new look at local average treatment effects with or without defiers},
	author={Dahl, Christian M and Huber, Martin and Mellace, Giovanni},
	journal={The Econometrics Journal},
	volume={26},
	number={3},
	pages={378--404},
	year={2023},
	publisher={Oxford University Press}
}

@article{richardson2010analysis,
	title={Analysis of the binary instrumental variable model},
	author={Richardson, Thomas S and Robins, James M},
	journal={Heuristics, Probability and Causality: A Tribute to Judea Pearl},
	volume={25},
	pages={415--444},
	year={2010},
	publisher={Citeseer}
}

@article{goff2020vector,
	title={A vector monotonicity assumption for multiple instruments},
	author={Goff, Leonard},
	journal={arXiv preprint arXiv:2009.00553},
	year={2020}
}

@article{kazemi2024instrumental,
	title={Instrumental variable analysis with categorical treatment},
	author={Kazemi, Amir Aamodt and Olsen, Inge Christoffer},
	journal={Statistical Methods in Medical Research},
	volume={33},
	number={11-12},
	pages={2043--2061},
	year={2024},
	publisher={SAGE Publications Sage UK: London, England}
}

@book{van2023limited,
	title={Limited Monotonicity and the Combined Compliers LATE},
	author={van't Hoff, Nadja and Lewbel, Arthur and Mellace, Giovanni},
	year={2023},
	publisher={University of Southern Denmark, Faculty of Business and Social Sciences~…}
}

@article{fusejima2024identification,
	title={Identification of multi-valued treatment effects with unobserved heterogeneity},
	author={Fusejima, Koki},
	journal={Journal of Econometrics},
	volume={238},
	number={1},
	pages={105563},
	year={2024},
	publisher={Elsevier}
}

@article{bai2024sharp,
	title={Sharp Testable Implications of Encouragement Designs},
	author={Bai, Yuehao and Huang, Shunzhuang and Tabord-Meehan, Max},
	journal={arXiv preprint arXiv:2411.09808},
	year={2024}
}

@article{lee2020treatment,
	title={Treatment effects with targeting instruments},
	author={Lee, Sokbae and Salani{\'e}, Bernard},
	journal={arXiv preprint arXiv:2007.10432},
	year={2020}
}

@article{mogstad2021causal,
	title={The causal interpretation of two-stage least squares with multiple instrumental variables},
	author={Mogstad, Magne and Torgovitsky, Alexander and Walters, Christopher R},
	journal={American Economic Review},
	volume={111},
	number={11},
	pages={3663--3698},
	year={2021},
	publisher={American Economic Association 2014 Broadway, Suite 305, Nashville, TN 37203}
}

@article{kirkeboen2016field,
	title={Field of study, earnings, and self-selection},
	author={Kirkeboen, Lars J and Leuven, Edwin and Mogstad, Magne},
	journal={The quarterly journal of economics},
	volume={131},
	number={3},
	pages={1057--1111},
	year={2016},
	publisher={Oxford University Press}
}

@article{mountjoy2022community,
	title={Community colleges and upward mobility},
	author={Mountjoy, Jack},
	journal={American Economic Review},
	volume={112},
	number={8},
	pages={2580--2630},
	year={2022},
	publisher={American Economic Association 2014 Broadway, Suite 305, Nashville, TN 37203}
}

@article{kline2016evaluating,
	title={Evaluating public programs with close substitutes: The case of Head Start},
	author={Kline, Patrick and Walters, Christopher R},
	journal={The Quarterly Journal of Economics},
	volume={131},
	number={4},
	pages={1795--1848},
	year={2016},
	publisher={MIT Press}
}

@article{pinto2021beyond,
	title={Beyond intention to treat: Using the incentives in moving to opportunity to identify neighborhood effects},
	author={Pinto, Rodrigo},
	journal={NBER Working Paper},
	year={2021}
}

@article{goff2024does,
	title={When does IV identification not restrict outcomes?},
	author={Goff, Leonard},
	journal={arXiv preprint arXiv:2406.02835},
	year={2024}
}

@article{bai2024identifying,
	title={On the Identifying Power of Monotonicity for Average Treatment Effects},
	author={Bai, Yuehao and Huang, Shunzhuang and Moon, Sarah and Shaikh, Azeem M and Vytlacil, Edward J},
	journal={arXiv preprint arXiv:2405.14104},
	year={2024}
}

@article{bai2024inference,
	title={Inference for Treatment Effects Conditional on Generalized Principal Strata using Instrumental Variables},
	author={Bai, Yuehao and Huang, Shunzhuang and Moon, Sarah and Santos, Andres and Shaikh, Azeem M and Vytlacil, Edward J},
	journal={arXiv preprint arXiv:2411.05220},
	year={2024}
}

@article{vytlacil2002independence,
	title={Independence, monotonicity, and latent index models: An equivalence result},
	author={Vytlacil, Edward},
	journal={Econometrica},
	volume={70},
	number={1},
	pages={331--341},
	year={2002},
	publisher={JSTOR}
}

@article{balke1997bounds,
	title={Bounds on treatment effects from studies with imperfect compliance},
	author={Balke, Alexander and Pearl, Judea},
	journal={Journal of the American statistical Association},
	volume={92},
	number={439},
	pages={1171--1176},
	year={1997},
	publisher={Taylor \& Francis}
}

@article{huber2015testing,
	title={Testing instrument validity for LATE identification based on inequality moment constraints},
	author={Huber, Martin and Mellace, Giovanni},
	journal={Review of Economics and Statistics},
	volume={97},
	number={2},
	pages={398--411},
	year={2015},
	publisher={The MIT Press}
}

@article{kitagawa2015test,
	title={A test for instrument validity},
	author={Kitagawa, Toru},
	journal={Econometrica},
	volume={83},
	number={5},
	pages={2043--2063},
	year={2015},
	publisher={Wiley Online Library}
}

@article{mourifie2017testing,
	title={Testing local average treatment effect assumptions},
	author={Mourifi{\'e}, Ismael and Wan, Yuanyuan},
	journal={Review of Economics and Statistics},
	volume={99},
	number={2},
	pages={305--313},
	year={2017},
	publisher={MIT Press One Rogers Street, Cambridge, MA 02142-1209, USA journals-info~…}
}

@article{sun2023instrument,
	title={Instrument validity for heterogeneous causal effects},
	author={Sun, Zhenting},
	journal={Journal of Econometrics},
	volume={237},
	number={2},
	pages={105523},
	year={2023},
	publisher={Elsevier}
}

@article{kwon2024testing,
	title={Testing Mechanisms},
	author={Kwon, Soonwoo and Roth, Jonathan},
	journal={arXiv preprint arXiv:2404.11739},
	year={2024}
}

@article{kedagni2020generalized,
	title={Generalized instrumental inequalities: testing the instrumental variable independence assumption},
	author={K{\'e}dagni, D{\'e}sir{\'e} and Mourifi{\'e}, Ismael},
	journal={Biometrika},
	volume={107},
	number={3},
	pages={661--675},
	year={2020},
	publisher={Oxford University Press}
}

@article{kaido2025testing,
	title={Testing Exclusion and Shape Restrictions in Potential Outcomes Models},
	author={Kaido, Hiroaki and Ponomarev, Kirill},
	journal={arXiv preprint arXiv:2512.20851},
	year={2025}
}

@article{rose2024recoding,
	title={On recoding ordered treatments as binary indicators},
	author={Rose, Evan K and Shem-Tov, Yotam},
	journal={Review of Economics and Statistics},
	pages={1--32},
	year={2024},
	publisher={MIT Press 255 Main Street, 9th Floor, Cambridge, Massachusetts 02142, USA~…}
}

@article{yap2025sensitivity,
	title={Sensitivity of Policy-Relevant Treatment Parameters to Violations of Monotonicity},
	author={Yap, Luther},
	journal={Journal of Applied Econometrics},
	year={2025},
	publisher={Wiley Online Library}
}

@article{noack2021sensitivity,
	title={Sensitivity of LATE estimates to violations of the monotonicity assumption},
	author={Noack, Claudia},
	journal={arXiv preprint arXiv:2106.06421},
	year={2021}
}

@article{huber2014sensitivity,
	title={Sensitivity checks for the local average treatment effect},
	author={Huber, Martin},
	journal={Economics Letters},
	volume={123},
	number={2},
	pages={220--223},
	year={2014},
	publisher={Elsevier}
}

@article{galichon2011set,
	title={Set identification in models with multiple equilibria},
	author={Galichon, Alfred and Henry, Marc},
	journal={The Review of Economic Studies},
	volume={78},
	number={4},
	pages={1264--1298},
	year={2011},
	publisher={Oxford University Press}
}

@article{ekeland2010optimal,
	title={Optimal transportation and the falsifiability of incompletely specified economic models},
	author={Ekeland, Ivar and Galichon, Alfred and Henry, Marc},
	journal={Economic Theory},
	volume={42},
	pages={355--374},
	year={2010},
	publisher={Springer}
}

@article{beresteanu2012partial,
	title={Partial identification using random set theory},
	author={Beresteanu, Arie and Molchanov, Ilya and Molinari, Francesca},
	journal={Journal of Econometrics},
	volume={166},
	number={1},
	pages={17--32},
	year={2012},
	publisher={Elsevier}
}

@book{KleinbergTardos,
	author = {Kleinberg, Jon and Tardos, \'Eva},
	publisher = {Addison Wesley},
	title = {Algorithm Design},
	year = 2006
}
	\normalsize
	
	\onehalfspacing
	\appendix 
	
	\section{Proofs}\label{Appendix:Proofs}
	
	\subsection{Proofs for section \ref{Section:Binary}}\label{Appendix:ProofBinaryIV}

	\begin{prooftheorem}[Proof of theorem \ref{Theorem:SharpBinary}]\phantom{a}\\
		
		\underline{\textbf{1$\Longrightarrow$2:}}
		
		The proof follows immediately from proposition \ref{Proposition:UpperBoundOnFT}.
		
		\underline{\textbf{2$\Longrightarrow $1:}}
		
		Let $\boldsymbol{p}$ be a distribution over instrument response types such that:
		\begin{align}
			A_Z\boldsymbol{p}&\leq \boldsymbol{r}_Z\\
			A_X\boldsymbol{p}&=\boldsymbol{r}_X\\
			p\big(\tilde{G}(x_l,\mathcal{Z})\big)&\leq \Psi_{x_l,\mathcal{Z}}
		\end{align}
		
		The goal is to define a distribution over response types that satisfies equation \ref{Equation:Consistency} and proposition \ref{Proposition:SharpSelectionModel_LinearEquations}.
		
		To simplify exposition, let $\tilde{p}_{x_l}$ denote the quantity $p\big(\tilde{G}(x_l,\mathcal{Z})$.
		
		To this end, it is necessary to specify potential treatments and outcomes for every possible response types, including potential outcomes for types that never realize certain treatments and thus never realize the potential outcome associated to such treatment.
		
		Let $t\in\Delta(\mathcal{T})$ be a response types. Define the set of treatments that response type $t=(t^Z,t^X)$ realizes as $R(t)=\{x_l\in\mathcal{X}:(t^Z)_{k'}=x_k \textit{ for some } k=1,2\}$. Note that by construction $|R(t)|\leq 2$
		
		Let $\lambda$ be an arbitrary fixed distribution over $\mathcal{Y}$. Then, define $g':\Delta(\mathcal{T})\longrightarrow \mathbb{R}$ as follows:
		
		\begin{align}
			g'(t^Z,t^X)=\begin{cases}
				\Big[\frac{1}{\Psi_{x_l,\mathcal{Z}}}\psi_{x_l,\mathcal{Z}}\big((t^X)_l\big)\Big]\Big[\underset{l':x_{l'}\not\in R(t)}{\prod}\lambda\big((t^X)_{l'}\big)\Big] & \textit{ if } R(t)=\{x_l\}\\
				\Big[\underset{l:x_l\in R(t)}{\prod}\frac{1}{\Phi_{x_l}-\tilde{p}_{x_l}}\xi_{x_l}\big((t^X)_l\big)\Big]\Big[\underset{l':x_{l'}\not\in R(t)}{\prod}\lambda\big((t^X)_{l'}\big)\Big] & \textit{ if }|R(t)|=2 \textit{ and }\Phi_{x_l}-\tilde{p}_{x_l}>0\\
				0 & \textit{ otherwise}
			\end{cases}
		\end{align} 
		Where $\xi_{x_l}\big((t^X)_l\big)=\phi_{Y|x_l,z_k}((t^X)_l)-\frac{\tilde{p}_{x_l}}{\Psi_{x_l,\mathcal{Z}}}\psi_{x_l,\mathcal{Z}}\big((t^X)_l\big)$.
		
		Integrating the function $g'$ allows to construct the desired distribution over response types. To this end, recall the following result from measure theory:
		
		\begin{lemma}\label{Lemma:Slices}
			Any measurable set of response types $B_{\mathcal{T}}\in\mathcal{B}_{\mathcal{T}}$ can be partitioned into slices as follows:
			\begin{align}
				B_{\mathcal{T}}=\underset{t^Z\in\mathcal{T}_{Z}}{\bigcup}\{t^Z\}\times B_{T_X}^{t^Z}
			\end{align}
			Where $B_{T_X}^{t^Z}=B_{Y_0}^{t^Z}\times B_{Y_1}^{t^Z}\times...\times B_{Y_{L-1}}^{t^Z}\in\mathcal{Y}^{L}$ is potentially empty and measurable with respect to the product measure over the product of the $L$ spaces that correspond to the potential outcomes at each different treatment value $x_l\in\mathcal{X}$. 
			
			Moreover, each individual set $=B_{Y_l}^{t^Z}$ for $l=0,1...,L-1$ is measurable with respect to $\mu_Y$.
		\end{lemma}
		
		Then, using the lemma, define:
		\begin{align}
			g\Big(B_{\mathcal{T}}\Big)&=\sum\limits_{t^Z\in\mathcal{T}_Z}\Big[p(t^Z)\int\limits_{B_{T_X}^{t^Z}}g'(t^Z,t^X)d\mu_{T^X}\Big]
		\end{align}
		Where $\mu_{T^X}$ denotes the product measure $\mu_{T^X}=\mu_Y\otimes \mu_Y\otimes...\otimes \mu_Y$ over the space of potential treatment-responses. 
		\begin{align}
			\int\limits_{B_{T_X}^{t^Z}}g'(t^Z,t^X)d\mu_{T^X}=\int\limits_{B_{Y_0}^{t^Z}}\int\limits_{B_{Y_1}^{t^Z}}...\int\limits_{B_{Y_{K-1}}^{t^Z}}g'(t^Z,t^X)d\mu_Y\mu_Y...\mu_Y
		\end{align}
		
		It remains to prove that $g\Big(B_{\mathcal{T}}\Big)$ is a valid distribution over reponse types, this is, it is non-negative, adds up to 1, satisfies countable additivity, and is consistent with the observed joint distribution of outcomes, treatments and instruments.. 
		
		When $|R(t)|=1$, the function $g'(t^Z,t^X)$ is the product of non-negative functions, therefore, it must be non-negative itself. When $|R(t)|=2$ it is either the product of non-negative functions, or zero. Therefore, it must be non-negative. From here, it follows that $g\Big(B_{\mathcal{T}}\Big)$ is non-negative for every measurable $B_{\mathcal{T}}\in \mathcal{B}_{\mathcal{T}}$.  
		
		Moreover, note that:
		\begin{align}
			\int\limits_{\mathcal{Y}}\frac{p(t^Z)}{\Psi_{x_l,\mathcal{Z}}}\psi_{x_l,\mathcal{Z}}\big(y\big)d\mu_Y=p(t^Z)\\
			\int\limits_{\mathcal{Y}}\Big[\frac{p(t^Z)}{\Phi_{x_l}-\tilde{p}_{x_l}}\big(\phi_{Y|x_l,z_k}(y)-\frac{\tilde{p}_{x_l}}{\Psi_{x_l,\mathcal{Z}}}\psi_{x_l,\mathcal{Z}}\big((y\big)\big)\Big]=p(t^Z)
		\end{align}
		
		From here, which follows that:
		\begin{align}
			g(\mathcal{T})&=\sum\limits_{t^Z\in\mathcal{T}_Z}\Big[p(t^Z)\int\limits_{\mathcal{T}_X}g'(t^Z,t^X)d\mu_{T^X}\Big]\\
			&=\sum\limits_{t^Z\in\mathcal{T}_Z}\Big[p(t^Z)\Big]\\
			&=1
		\end{align}
		Countable additivity follows mechanically. Let $\{B_{\mathcal{T}}^m\}_{m\in M}$ be a countable collection of disjoint subsets of $\mathcal{T}$ indexed by $m\in M$. Applying lemma \ref{Lemma:Slices}, each sets $B_{\mathcal{T}}^m$ can be sliced as follows:
		\begin{align}
			B_{\mathcal{T}}^m=\underset{t^Z\in\mathcal{T}_{Z}}{\bigcup}\{t^Z\}\times B_{T_X}^{t^Z,m}
		\end{align}
		Since $B^m\cap B^{m'}=\emptyset$ for all $m,m'\in M$ such that $m\neq m'$ then its must be that $B_{T_X}^{t^Z,m}\cap B_{T_X}^{t^Z,m'}=\emptyset$. 
		
		Then:
		\begin{align}
			g\Big(\underset{m\in M}{\bigcup}B_{\mathcal{T}}^m\Big)&=g\Big(\underset{m\in M}{\bigcup} \Big[\underset{t^Z\in\mathcal{T}_{Z}}{\bigcup}\Big\{\{t^Z\}\times B_{T_X}^{t^Z,m}\Big\}\Big]\Big)\\
			&=g\Big(\underset{t^Z\in\mathcal{T}_{Z}}{\bigcup} \Big[\underset{m\in M}{\bigcup}\Big\{\{t^Z\}\times B_{T_X}^{t^Z,m}\Big\}\Big]\Big)\\
			&=g\Big(\underset{t^Z\in\mathcal{T}_{Z}}{\bigcup} \Big[\{t^Z\}\times \Big\{\underset{m\in M}{\bigcup}B_{T_X}^{t^Z,m}\Big\}\Big]\Big)\\
			&=\sum\limits_{t^Z\in\mathcal{T}_Z}p(t^Z)\int\limits_{\underset{m\in M}{\bigcup}B_{T_X}^{t^Z,m}}g'(t^Z,t^X)d\mu_{T_X}\\
			&=\sum\limits_{t^Z\in\mathcal{T}_Z}\sum_{m\in M}p(t^Z)\int\limits_{B_{T_X}^{t^Z,m}}g'(t^Z,t^X)d\mu_{T_X}\\
			&=\sum_{m\in M}\sum\limits_{t^Z\in\mathcal{T}_Z}p(t^Z)\int\limits_{B_{T_X}^{t^Z,m}}g'(t^Z,t^X)d\mu_{T_X}\\
			&=\sum_{m\in M}g(B_{\mathcal{T}}^{m})
		\end{align}
		
		The first line is comes from application of lemma \ref{Lemma:Slices} to every set $B_{\mathcal{T}}^m$, the second simply inverts the order of the unions, the third is a property of cartesian products, the fourth simply applies the definition of $g$ and takes advantage of the fact that the set is already sliced, the fifth line uses countable additivity of the integral, the sixth changes the order of the sums, and the last one uses the definition of $g$ but now applied to each set $B_Y^m$ separately. This proves countable additivity.
		
		It only remains to verify consistency with observed probabilities, let $B_{T_X}^{t^Z}=\mathcal{Y}\times....\times B_Y\times...\times\mathcal{Y}$ be the set obtained by taking the product of the outcome's support $L-1$ times and letting the $l'-th$ element of the product coincide with $B_Y$ instead of $\mathcal{Y}$. Then:	
		
		\begin{align}
			g\Big(\big\{(t^Z,t^X)\in\mathcal{T}:(t^Z)_k=x_{l'},(t^X)_l\in B_Y\big\}\Big)&=\sum\limits_{t^Z:(t^Z)_k=x_l}\Big[p(t^Z)\int\limits_{B_{T_X}^{t^Z}}g'(t^Z,t^X)d\mu_{T^X}\Big]\\
			&=\frac{\tilde{p}_{x_l}}{\Psi_{x_l,\mathcal{Z}}}\int\limits_{B_Y}\psi_{x_l,\mathcal{Z}}\big(y\big)d\mu_Y\\
			&\phantom{aaa}+\frac{1-\tilde{p}_{x_l}}{1-\tilde{p}_{x_l}}\int\limits_{B_Y}\Big(\phi_{Y|x_l,z_k}(y)-\frac{\tilde{p}_{x_l}}{\Psi_{x_l,\mathcal{Z}}}\psi_{x_l,\mathcal{Z}}(y)\Big)d\mu_Y\\
			&=\int\limits_{B_Y}\phi_{Y|x_l,z_k}(y)d\mu_Y
		\end{align}
		The first equality is just the definition of $g$, the second uses the fact that at instrument value $z_k$ there are only two instrument-response types that realize $x_l$ (always-$x_l$ takers, and $\{z_k\}$-IFF-$x_l$-takers) combined with the definition of $g'$. The last line simply re-arranges terms.
		
		This completes the proof. 
		
	\end{prooftheorem}
	
	\subsubsection{Proofs for subsection \ref{Subsection:OTandFOSD}}\label{Appendix:ProofFOSD}
	
	Before providing the proofs I formally define the transportation problem as well as the maximum flow and minimum cut problems.
	
	Throught the definitions and proofs I alternate from edges of the graphs that define the transportation problem and response types. This is possible because there is a one-to-one correspondence between the response types $t^Z$ that belong to $\mathcal{T}_Z\setminus R$ and edges $e\in\mathcal{E}_R$. This correspondence is given by $t^Z=(x_l,x_{l'})\longrightarrow e=(x_l^0,x_{l'}^1)$. I denote the unique edge associated to $t^Z\in\mathcal{T}_Z$ as $e_{t^Z}$ and the unique instrument-response type associated to $e\in\mathcal{E}_R$ as $t^Z_{e}$.
	
	\begin{definition}[\textbf{Circulation with demands problem associated to R }]\label{Definition:CirculationProblem} Here I introduce a version of the problem taylored to the instrument validity application. For a detailed discussion of discrete circulation and transportation problems see \cite{KleinbergTardos} chapter 7.5.
		
		Let $R\subset\mathcal{T}_Z$ be a collection of instrument-response types with associated graph $\mathcal{E}_R$.
		
		Define demands $d_{x_l^k}$ for $x_l^k\in \mathcal{X}_0\cup\mathcal{X}_1$ as follows:
		\begin{align}
			d_{(x,z)}=\begin{cases}
				-\mathbb{P}[X_i=x_l|Z_i=z_k] & \textit{ if } k=0\\
				\mathbb{P}[X_i=x_l|Z_i=z_k]  & \textit{ if } k=1
			\end{cases}
		\end{align}
		
		Furthermore, every edge $e=(x_l^0,x_{l'}^1)\in\mathcal{E}_R$ has a maximum capacity $c_e\in[0,1]$ defined as $c_e=c_e((x_l^0,x_{l'}^1))$.

		A candidate solution to the circulation with capacities problem is a function $\tilde{p}:E_{R}\longrightarrow\mathbb{R}_{+}$ such that, for each directed edge $e=(x_l^{0},x_{l'}^1)\in\mathcal{E}_R$, the number $\tilde{p}(e)$, represents the flow from node $x_l^0$ to node $x_{l'}^1$. The flow $\tilde{p}$ solves the problem if:
		\begin{enumerate}
			\item Capacity constraints are satisfied: 
			\begin{align}
				0\leq \tilde{p}(e) \leq c_e \textit{ for all } e\in E_{R}
			\end{align}
			\item Demand is satisfied at every node (note that supply is represented as negative demand):
			\begin{align}
				\phantom{aaaaaaaaaa}d_{x_{l}^0}&=-\sum\limits_{e\in E^{out}(x_l^0)}\tilde{p}(e) &\textit{ for all } l=1,2,...,L\\
				\phantom{aaaaaaaaaa}d_{x_{l'}^1}&=\sum\limits_{e\in E^{in}(x_{l'}^{1})}\tilde{p}(e) &\textit{ for all } l=1,2,...,L
			\end{align}	
		\end{enumerate}
		Where $E^{out}(x_l^0)=\{e\in\mathcal{E}_R:e=(x_l^{0},x_{l'}^1) \textit{ for some } x_{l'}^1\in\mathcal{X}_1\}$ and $E^{in}(x_{l'}^{1})=\{e\in\mathcal{E}_R:e=(x_l^{0},x_{l'}^1) \textit{ for some } x_{l}^1\in\mathcal{X}_0\}$ represent the sets of edges that depart from $x_l^0$ and end in $x_{l'}^1$ respectively.
	\end{definition}
	\medskip
	\begin{definition}[\textbf{Maximum flow problem associated to R}]\label{Definition:MaxFlowProblem}
		Let $R\subset\mathcal{T}_Z$ be a collection of instrument-response types with associated graph $\mathcal{E}_R$. Let $(\mathcal{X}_0\cup\mathcal{X}_1,\mathcal{E}_R)$ be the graph from definition \ref{Definition:CirculationProblem}. From theorem 7.50 in \cite{KleinbergTardos} that any circulation with demands problem is equivalent to a maximum flow problem. The maximum flow problem associated to $R$ is obtained by applying this principle to the associated circulation with demands problem.
		
		The directed graph for the maximum flow problem is defined as follows:
		
		Two vertices are added: a source $s^*$ and a sink $t^*$ to obtain the set of vertices $\mathcal{X}_0\cup\mathcal{X}_1\cup\{s^*,t^*\}$.
		
		A directed edge points from the source to every node $x_l^1$ and from every node $x_{l'}^1$ to the sink:
		\begin{align}
			\mathcal{E}_R^{MaxFlow}=\mathcal{E}_R\bigcup \Big\{(x_{l}^0,t^*):x\in\mathcal{X}\Big\}\bigcup \Big\{ (s^*,x_{l'}^1\ ):x\in\mathcal{X}\Big\}
		\end{align}
		
		The capacity of all edges of the form $e=(x_{l}^0,t^*)$ is set to $c_e=\mathbb{P}[X_i=x_l|Z_i=z_0]$, the capacity of all edges of the form $e=(s^*,x_{l'}^1\ )$ is set to $c_e=\mathbb{P}[X_i=x_{l'}|Z_i=z_1]$. All remaining edges have the same capacities as in definition \ref{Definition:CirculationProblem}. 
		
		A flow is a function $\tilde{p}:\mathcal{E}_R^{MaxFlow}\longrightarrow\mathbb{R}_+$ such that every intermediate node (i.e. all nodes except $s^*$ and $t^*$ has exactly the same inflow and outflow). Formally:
		
		\begin{enumerate}
			\item For every $x_l^0\in\mathcal{X}_0$:
			\begin{align}
				\sum\limits_{e'\in E^{out}(x_l^0)}\tilde{p}(e')=\tilde{p}\big((s^*,x_{l}^0)\big)
			\end{align}
			\item For every $x_{l'}^1\in\mathcal{X}_1$:
			\begin{align}
				\sum\limits_{e'\in E^{in}(x_{l'}^1)}\tilde{p}(e')=\tilde{p}\big((x_{l'}^1,t^*)\big)
			\end{align}
		\end{enumerate}
		Let $\tilde{P}$ denote the set of all flows. The problem is:
		\begin{align}
			\underset{\tilde{p}\in\tilde{P}}{\textit{max.   }}\sum\limits_{x_l^0\in\mathcal{X}_0}\tilde{p}\big((s^*,x_{l}^0)\big)\\
			\textit{subject to: } \tilde{p}(e)\leq c_e. 
		\end{align}
	\end{definition}

	\begin{definition}[\textbf{$s^*-t^*$ cut of $\big(\mathcal{X}_0\cup\mathcal{X}_1\cup\{s^*,t^*\},\mathcal{E}_R^{MaxFlow}\big)$}]\label{Definition:Cut}
		Let $\big(\mathcal{X}_0\cup\mathcal{X}_1\cup\{s^*,t^*\},\mathcal{E}_R^{MaxFlow}\big)$ be the graph from definition \ref{Definition:MaxFlowProblem}. 
		
		\begin{enumerate}
			\item A cut $C=(S_C,T_C)$ is a partition (with elements $S$ and $T$) of the set $\mathcal{X}_0\cup\mathcal{X}_1\cup\{s^*,t^*\}$ such that $s^*\in S_C$ and $t^*\in T_C$. I denote the set of all cuts of $\big(\mathcal{X}_0\cup\mathcal{X}_1\cup\{s^*,t^*\},\mathcal{E}_R^{MaxFlow}\big)$ as $\mathcal{C}(R,\mathcal{X})$.
			\item The cut set of $C=(S_C,T_C)$ denoted $\Xi_C$, is defined as the set of edges that, if removed, would disconect $S_C$ from $T_C$, formally: 
			\begin{align}
				\Xi_C=\{e=(e_1,e_2)\in \mathcal{E}_R^{MaxFlow}|e_1\in S_C \textit{ and } e_2\in T_C \}
			\end{align}
			\item The capacity of a cut $C$, denoted $Cap(C)$, is defined as the sum of the capacities of the edges in $\Xi_C$:
			\begin{align}
				Cap(C)=\sum_{e\in\Xi_C}c_e
			\end{align} 
			\item The minimum cut associated to $R$ is $\underset{C\in \mathcal{C}(R,\mathcal{X}) }{max}\textit{  }Cap(C)$
		\end{enumerate}

	\end{definition}	
	
	\begin{proposition}\label{Proposition:MinCutEquivalence}
		The following statements are equivalent: 
		\begin{enumerate}
			\item A distribution over instrument response types exists that satisfies assumptions (\ref{Assumption:Validity}), (\ref{Assumption:IVResponseTypeRestriction}) and $p(x_l,x_{l'})\leq c_e=c_e((x_l^0,x_{l'}^1))$ exists. 
			\item A solution to the circulation with demands problem from definition \ref{Definition:CirculationProblem} exists.
			\item The maximum flow attainable in the problem from definition \ref{Definition:MaxFlowProblem} is equal to 1.
			\item The minimum cut attainable in  definition \ref{Definition:Cut} is equal to 1. 
		\end{enumerate}
	\end{proposition}
	
	\begin{proof}[\textbf{Proof of proposition \ref{Proposition:MinCutEquivalence}}]
		I prove the equivalence between $1$ and $2$. The equivalence of 2, 3 and 4 follows from well known results in operations research (See \cite{KleinbergTardos} theorems 7.9 and 7.50). %
		
	for reference. 
	
	\underline{\textbf{$2\Longrightarrow 1:$}}\\
	Let $\tilde{p}$ be a flow that satisfies the circulation with demands problem associated to $R$.\\
	\\
	For any instrument-response type $t^Z=(x_l,x_{l'})\in\mathcal{T}_Z$ define the distribution over instrument-response types $p$ as follows:
	\begin{align}
		p(t^Z)= p\big((x_l,x_{l'})\big):=\begin{cases} \tilde{p}(e_{t^Z})& \textit{ if }e_{t^Z}\in\mathcal{E}_R\\
			0 & otherwise
		\end{cases}
	\end{align}
	
	Consider any node $x_l^{0}$:
	\begin{align}
		-\mathbb{P}[X_i=x|Z_i=0]&=d_{x_l^{0}}\\
		&=-\sum\limits_{e\in E^{out}(x_l^{0})}\tilde{p}(e)\\
		&=-\sum_{x_{l'}\in\mathcal{X}}p((x_l,x_{l'}))
	\end{align}
	The first two equalities follow from the definition of the transportation problem, the last from $p(t^Z)=\tilde{p}(e_{t^Z})$.\\
	An analogous argument shows that for $z=1$:
	\begin{align}
		\mathbb{P}[X_i=x_l|Z_i=1z_]=\sum_{x'\in\mathcal{X}}p((x_{l'},x_l))
	\end{align}	
	Which shows that $p$ is consistent with observed conditional treatment probabilities.
	
	Capacity constraints are trivially satisfied since $\tilde{p}(e_{t^Z})\leq c_e=c_e((x_l^0,x_{l'}^1))$.
	
	Finally, from the definition of $\mathcal{E}_R$ it follows that $\tilde{p}$ only assigns non-negative flows to edges in $e\in\mathcal{E}_R$, therefore, $p$ only assigns non-negative probabilities to response types $t^{Z}\in R$. This shows consistency with assumption \ref{Assumption:IVResponseTypeRestriction}. 
	
	\underline{\textbf{$1\Longrightarrow 2:$}}\\
	Let $p$ be a distribution over response types that satisfies assumptions (\ref{Assumption:Validity}) and (\ref{Assumption:IVResponseTypeRestriction}). \\
	For every edge $e=(x_l^{0},x_{l'}^1)$ define the flow $\tilde{p}$ as follows:
	\begin{align}
		\tilde{p}(e)=\tilde{p}\big((x_l^{0},x_{l'}^1)\big)=p(x_l,x_{l'})
	\end{align}
	Note that $e$ is always well defined since every response type is assigned non-negative probability. 
	
	Now, consider an arbitrary node $x_l^0\in\mathcal{X}_0$:
	\begin{align}
		d_{x_l^0}=&-\mathbb{P}[X_i=x_l|Z_i=z_0]\\
		&=-\sum_{x_{l'}\in\mathcal{X}}p((x_l,x_{l'}))\\
		&=-\sum\limits_{e\in E^{out}(x_l^{0})}\tilde{p}(e)
	\end{align}
	An analogous argument shows that for $x_{l'}^1$:
	\begin{align}
		d_{(x,z)}=\sum\limits_{e\in E^{in}(x_{l'}^{1})}\tilde{p}(e)
	\end{align}	
	Thus, demand is satisfied at every node. Again, capacity constraints are trivially satisfied since $\tilde{p}((x_l^0,x_{l'}^1))\leq c_e=c_e((x_l^0,x_{l'}^1))$. This completes the proof.
	
\end{proof}

\begin{proof}[\textbf{Proof of theorem \ref{Theorem:FOSD}}]\phantom{A}\\
	\textbf{Part 1:} First note that the restriction $c_{t^Z}=1$ simply states that instrument-response probabilities are not bounded from above other than by the requirement that they are less than 1. This is, response types in $R$ can occur with arbitrary frequency in the population.
	
	\underline{$\Longrightarrow$} 
	\begin{align}
		\mathbb{P}[X_i\in S|Z_i=z_0]&=p(\{t^Z\in\mathcal{T}^Z: t^Z=(x_l,x_{l'})\in S\times\mathcal{X}\} )\\
		&=\sum\limits_{x_{l'}\in\mathcal{X}}p(\{t^Z\in\mathcal{T}^Z:t^Z=(x_l,x_{l'})\in S\times\{x_{l'}\}\})\\
		&=\sum\limits_{x_{l'}\in\{x_{l'}\in\mathcal{X}: (s,x_{l'})\in R \textit { and } s\in S  \}} (\{t^Z\in\mathcal{T}^Z:t^Z=(x_l,x_{l'})\in S\times\{x_{l'}\}\})\\
		&=\sum\limits_{x_{l'}\in L_{\geq^R}(S)^C} p(\{t^Z\in\mathcal{T}^Z:t^Z=(x_l,x_{l'})\in S\times\{x_{l'}\}\})\label{AuxEcuation:FOSD} \\
		&\leq \sum\limits_{x_{l'}\in L_{\geq^R}(S)^C}\mathbb{P}[X_i=x_{l'}|Z_i=1]\\
		&=\mathbb{P}[X_i\in L_{\geq^R}(S)^C | Z_i=z_1]
	\end{align}
	The first equality comes from the consistency of $p$ with observed probabilities, the second from decomposing the probability in terms of all the different response types that realize $x_l$ when $Z=0$, the third equality comes from the fact that $p$ assigns zero probability to response types outside of $R$, and the fourth noting that if $(x_l,x_{l'})\in R$ then $\not x_{l}\geq^R x_{l'}$ which implies $x_{l'}\not\in L_{\geq^R}(S)$. The inequality comes from noting that $p$ is consistent with observed probabilities and $\{x_{l'}\}$ could also be realized when $Z=1$ by types whose potential outcome when $Z=0$ lies outside of $S$. The last equality comes from adding disjoint probabilities across the elements of $L_{\geq^R}(S)^C$. 
	
	\underline{$\Longleftarrow$}
	
	I prove the counterpositive, this is, if no $p$ that satisfies the hypotheses exists, then, a violation of equation \ref{Equation:FOSD} must exist. To this end I use the minimum cut characterization. 
	
	First, note that by setting $C=\{e\in\mathcal{E}_R:e=(s^*,x_l^0)\textit{ for some }x_l\in\mathcal{X}_0\}$ a cut of capacity one can always be attained. Therefore, edges of the form $(x_l^0,x_{l'}^1)$ will never be part of a minimum cut since their capacity is 1 and removing a single edge will not induce a cut. Thus, minimum cut sets will only contain edges of the form $(s^*,x_l^0)$ or $(x_{l'}^1,t^*)$. 
	
	Suppose that $C$ is a minimum cut such that $Cap(C)\leq 1$. Let $S$ be the set of all $x_{l}\in\mathcal{X}$ such that the edge $(s^*,x_{l}^0)$ belongs to $\Xi_C$. Then, by construction, any node of the form $(s^*,x_{l'}^0)$ for $x_{l'}\in S^C=\mathcal{X}\setminus S'$ does not belong to the cut. 
	
	Define $S':=\mathcal{X}\setminus S$ and note that by construction for any $x_{l'}\in S'$ and any $x_{l''}\in \mathcal{X}$ such that $(x_{l'},x_{l''})\in R$, the edge $(x_{l''}^1,t^*)$ belongs to $C$. Otherwise, the path $(s^*,x_{l'}),(x_{l'},x_{l''}),(x_{l''},t^*)$ would connect $t^*$ and $s^*$.
	
	But the set of elements $x_{l''}\in \mathcal{X}$ such that $(x_{l'},x_{l''})\in R$ for some $x_{l'}\in S'$ is actually the set $L_{\geq^R}(S')^C$.
	
	Therefore:
	\begin{align}
		\sum\limits_{c_e\in\Xi_C}c_e&<1\\
		\Longrightarrow\sum\limits_{e\in\mathcal{E}_R\cap\{s^*\}\times S}c_e+\sum\limits_{e\in\mathcal{E}_R\cap L_{\geq^R}(S')\times \{t^*\}}c_e&<1\\
		\Longrightarrow\phantom{aaaaaaaaaaaaaa}\sum\limits_{e\in\mathcal{E}_R\cap L_{\geq^R}(S')\times \{t^*\}}c_e&<1-\sum\limits_{e\in\mathcal{E}_R\cap\{s^*\times S\}}c_e\\
		\Longrightarrow\phantom{aaaaaaaaaaaaaa}\sum\limits_{e\in\mathcal{E}_R\cap  L_{\geq^R}(S')\times \{t^*\}}c_e&<\sum\limits_{e\in\mathcal{E}_R\cap\{s^*\}\times S'}c_e\\
		\Longrightarrow \phantom{aaaaaaaaaaaai}\mathbb{P}[X_i\in L_{\geq^R}(S')|Z_i=z_1]&<\mathbb{P}[X_i\in S'|Z_i=z_0]
	\end{align}
	The first implication follows from the previous discussion. The second from re-arranging terms. The third from $S'=S^C$ and the fourth from the definition of the capacities. The last equation is a violation of equation \ref{Equation:FOSD}.
	
	\textbf{Part 2:} 
	\underline{$\Longrightarrow$:} Suppose that a distribution over instrument-response types $p\in\Delta(\mathcal{T}_Z)$ exists that satisfies assumptions $1$ and $3$.
	
	For $\emptyset\neq S\subset\mathcal{X}$, let $(\Lambda_S,\Lambda'_S)$ be a partition of $L_{\geq^R}(S)^C$ such that the following three conditions are satisfied: 
	\begin{enumerate}
		\item $\emptyset\neq\Lambda'_S\subset S$.
		\item For every $x_{l'}\in \Lambda'_S$ the set $L_{\geq^R}(x_{l'})^C\setminus\{x_{l'}\}$ is contained in $\Lambda_S$.
		\item $L_{\geq^R}(S\setminus \Lambda_S')^C\subset \Lambda_S$.
	\end{enumerate}
	
	Take equation \ref{AuxEcuation:FOSD}:
	\begin{align}
		\mathbb{P}[X_i\in S|Z_i=z_0]=&\sum\limits_{x_{l'}\in L_{\geq^R}(S)^C} p(\{t^Z\in\mathcal{T}^Z:t^Z=(x_l,x_{l'})\in S\times\{x_{l'}\}\}) \\
		=&\sum\limits_{x_{l'}\in \Lambda} p(\{t^Z\in\mathcal{T}^Z:t^Z=(x_l,x_{l'})\in S\times\{x_{l'}\}\})\\
		&\phantom{aaaaaaa}+\sum\limits_{x_{l'}\in \Lambda'} p(\{t^Z\in\mathcal{T}^Z:t^Z=(x_l,x_{l'})\in S\times\{x_{l'}\}\}) \nonumber\\
		=&\sum\limits_{x_{l'}\in \Lambda} p(\{t^Z\in\mathcal{T}^Z:t^Z=(x_l,x_{l'})\in S\times\{x_{l'}\}\})\\
		&\phantom{aaaaaaa}+\sum\limits_{x_{l'}\in \Lambda'} p(\{t^Z\in\mathcal{T}^Z:t^Z=(x_l,x_{l'})\in (S\setminus \Lambda_S')\times\{x_{l'}\}\}) \nonumber\\
		&\phantom{aaaaaaa}+\sum\limits_{x_{l'}\in \Lambda'} p(\{t^Z\in\mathcal{T}^Z:t^Z=(x_l,x_{l'})\in \Lambda_S'\times\{x_{l'}\}\}) \nonumber
	\end{align}
	The second term from the last equation is equal to zero by condition 3. Condition 2 implies that the only possible value that $x_l$ can take in the last term is $x_{l'}$. Therefore, the equation becomes:
	\begin{align}
		\mathbb{P}[X_i\in S|Z_i=z_0]=&\sum\limits_{x_{l'}\in \Lambda} p(\{t^Z\in\mathcal{T}^Z:t^Z=(x_l,x_{l'})\in S\times\{x_{l'}\}\})\\
		&\phantom{aaaaaaa}+\sum\limits_{x_{l'}\in \Lambda'} p(\{t^Z\in\mathcal{T}^Z:t^Z=(x_l,x_{l})\textit{ such that } x_l\in\Lambda'\}) \nonumber\\
		\leq& \mathbb{P}[X_i\in\Lambda|Z_i=z_1]+\sum\limits_{x_l\in\Lambda'} p((x_l,x_l))\\
		\leq& \mathbb{P}[X_i\in\Lambda|Z_i=z_1]+\sum\limits_{x_l\in\Lambda'} \Psi_{x_l}
	\end{align}
	The first inequality follows from consistency of $p$ with observed probabilities and the fact that response types that realize treatments outside of $S$ when $Z=0$ might also realize treatments in $\Lambda$ when $Z=1$ along with the fact that $L_{\geq^R}(x_{l})^C\setminus\{x_{l}\}\subset\Lambda$ for $x_l\in\Lambda'$. The second inequality follows from theorem \ref{Theorem:SharpBinary}.
	
	\underline{$\Longleftarrow$} 
	
	As for part 1, I prove the contrapositive statement. That is, I suppose that a cut with capacity less than 1 exists and show that it implies that inequality (\ref{Equation:GeneralizedFOSD}) from theorem (\ref{Theorem:FOSD}) is not satisfied for some $S\subset\mathcal{X}$.
	
	Now the capacities of edges of the form $(x_l^0,x_l^1)$ are potentially binding. For cuts that only contain edges connected to $s^*$ or $t^*$ the argument from part 1 still applies. Edges of the form $(x_l^0,x_{l'}^1)$ for $x_l\neq x_{l'}$ still have capacity 1 and therefore will not belong to a minimum cut that is strictly less than one. 
	
	Suppose that $C$ is a minimum cut with at least one element of the form $(x_l^0,x_l^1)$.
	
	Let $S'$ be the set of elements in $x_l\in\mathcal{X}$ such that $(s^*,x_l^0)\in C$. Define $S=\mathcal{X}\setminus S'$. Note that for any $x_{l}\in S$ and any $x_{l'}\in L_{\geq^R}(x_l)^C$ one of the following sets of conditions must hold:
	
	\begin{enumerate}
		\item The following holds:
		\begin{itemize}
			\item $(x_{l'}^1,t^*)\in C$
		\end{itemize}
		\item The following three conditions hold
		\begin{itemize}
			\item $x_l=x_{l'}\in S\cap L_{\geq^R}(S)^C$
			\item $(x_{l'}^0,x_{l'}^1)\in C$.
			\item $(x_{l''}^1,t^*)\in C$ for all $x_{l''}\in L_{\geq^R}(x_{l'})^C\setminus\{x_{l'}\}$
		\end{itemize}
	\end{enumerate}

	The first condition prevents paths of the form $(s^*,x_{l}^0),(x_{l}^0,x_{l'}^1),(x_{l'}^1,t^*)$ as in part 1. 
	
	If condition 1 is not satisfied and $(x_{l'}^1,t^*)\not\in C$, then, the path $(s^*,x_{l}^0),(x_{l}^0,x_{l'}^1),(x_{l'}^1,t^*)$ must be cut at the middle at edge $(x_{l}^0,x_{l'}^1)$. Because edges associated to $x_l\neq x_{l'}$ cannot belong to the cut, it must be that $x_l=x_{l'}\in S'\cap L_{\geq^R}(S')^C$. Furthermore, for every $x_{l''}\in L_{\geq^R}(x_l)^C\setminus\{x_l\}$, the path $(s^*,x_l^0),(x_l^0,x_{l''}^1),(x_{l''}^1,t^*)$ must be blocked at some point. By construction, this can only happen at $(x_{l''}^1,t^*)$, thus, $(x_{l''}^1,t^*)\in C$. 
	
	In other words, a path $(s^*,x_{l}^0),(x_{l}^0,x_{l'}^1),(x_{l'}^1,t^*)$ might be blocked by condition 1, if not, then, conditions 2 must hold.
	
	Let $\Lambda'_S$ be the set $\Lambda'_S=\{x_l\in\mathcal{X}: (x_l^0,x_l^1)\in C\}\subset  L_{\geq^R}(S)^C\cap S $. Set $\Lambda_S=L_{\geq^R}(S)^C\setminus \Lambda'$. By construction the sets $\Lambda'_S$ and $\Lambda_S$ partition $L_{\geq^R}(S)^C$.
	
	Note that $\{x_{l'}\in\mathcal{X}:(x_{l'},t^*)\in C\}$ must be a subset of $L_{\geq^R}(S)^C$, otherwise, $x_{l'}^1$ would not be reachable from $\{x_l^0:x_l\in S\}$ and eliminating $(x_{l'},t^*)$ from $C$ would attain a cut with strictly lower capacity. Thus, $\{x_{l'}\in\mathcal{X}:(x_{l'},t^*)\in C\}\subset\Lambda_S$.
	
	Similarly, note that if $x_{l'}\in\Lambda_S$ it must be that $x_l\in \{x_{l'}\in\mathcal{X}:(x_{l'},t^*)\in C\}$, otherwise, a path through some $x_l\in S$ such that $L_{\geq^R}(x_l)^C\cap \Lambda_S\neq\emptyset$ to $x_{l'}$ would connect $s^*$ to $t^*$. Hence: 
	\begin{align}
		\{x_{l'}\in\mathcal{X}:(x_{l'},t^*)\in C\}=\Lambda \label{Equation:SetEqualityForFOSD}
	\end{align}
	
	Finally, suppose that for some $x_l\in S\setminus\Lambda'_S$ exists some $x_{l'}\in L_{\geq^R}(x_l)^C\cap \Lambda_S'\neq\emptyset$. By definition it must be that $x_{l'}\neq x_l$. Then, the path $(s^*,x_l),(x_l,x_{l'}),(x_{l'},t^*)$ would connect $s^*$ to $t^*$. Hence, it must be that:
	\begin{align}
		L_{\geq^R}(S\setminus \Lambda'_S)^C\subset\Lambda_S \label{Equation:ContentionForFOSD}
	\end{align}

	Then:
	\begin{itemize}
		\item $\emptyset\neq \Lambda' \subset S$ (By conditions 2).
		\item For all $x_{l'}\in\Lambda'$, $L_{\geq^R}(x_l')^C\setminus\{x_{l'}\}\subset\Lambda$ (By conditions 2 and equation (\ref{Equation:SetEqualityForFOSD})).
		\item $L_{\geq^R}(S\setminus \Lambda'_S)^C\subset\Lambda_S$ by expression (\ref{Equation:ContentionForFOSD}).
	\end{itemize}

	Moreover, it must be that:
	\begin{align}
		\sum\limits_{c_e\in\Xi_C}c_e&< 1 \\
		\Longrightarrow \sum\limits_{c\in C\cap \{s^* \}\times S'}c_e+\sum\limits_{\{c\in C: c=(x_l^0,x_l^1)\textit{ s.t. } x_l\in\Lambda'\}}c_e+\sum\limits_{\{c\in C: c=(x_l^1,t^*)\textit{ s.t } x_l\in\Lambda\}}c_e&< 1\\
		\Longrightarrow \phantom{aaaaaaaaaaaaa} \sum\limits_{\{c\in C: c=(x_l^0,x_l^1)\textit{ s.t. } x_l\in\Lambda'\}}c_e+\sum\limits_{\{c\in C: c=(x_l^1,t^*)\textit{ s.t } x_l\in\Lambda\}}c_e&<1- \sum\limits_{c\in C\cap \{s^* \}\times S'}c_e\\
		\Longrightarrow \phantom{aaaaaaaaaaaaaaaaaaaa} \sum\limits_{\{c\in C: c=(x_l^0,x_l^1)\textit{ s.t. } x_l\in\Lambda'\}}c_e+\mathbb{P}[X_i\in\Lambda|Z=1]&< \mathbb{P}[X_i\in S|Z=z_0]
	\end{align}
	Where the first implication follows from the previous discussion, the second from rearranging terms and the last from the definition of the capacities.	
	
	Setting $c_e((x_l^0,x_{l'}^1))=\Psi_{x_l}$ completes the proof.
	
\end{proof}

\subsubsection{Proofs for subsection \ref{Subsection:Submonotonicity}} \label{Appendix:ProofSubmonotonicity}

To prove theorem \ref{Theorem:Submonotonicity} I begin by formally defining latent preference types. To this end, let $\mathcal{L}_s(\mathcal{X})$ denote the set of all strict preference relations over $\mathcal{X}$ and let $\mathcal{N}(x_l)=\{\succ\in\mathcal{L}_s(\mathcal{X}): x_l\succ x_{l'}\textit{ for all } x_{l'}\in\mathcal{X}\}$ denote the subset of strict preference relations that has $x_l$ as a maximizer.

\begin{definition}[Latent preference types]
	A latent preference type $t^{\succ}\in\mathcal{T}_{\succ}$ is a function that maps instrument realizations to strict preference profiles: $t^{\succ}:\mathcal{Z}\longrightarrow \mathcal{L}_s(\mathcal{X})$. 
	
	Note that $t^{\succ}(z)$ is a preference relation. Therefore, the equation $x_l t^{\succ}(z) x_{l'}$ reads as ``latent preference types $t^{\succ}$ strictly prefer $x_l$ over $x_{l'}$ when $Z_i=z$''. To improve readibility, I denote $t^{\succ}(z)$ as $\succ^{t(z)}$ so that the previous equation becomes  $x_l \succ^{t(z)} x_{l'}$.
\end{definition}

\begin{definition}[Distribution over latent preference types]
	A distribution over latent preference types is a function $h:\mathcal{H}\longrightarrow [0,1]$ that satisfies:
	\begin{align}
		h(t^{\succ})\geq 0 \textit{ for all } t^{\succ}\in\mathcal{T}_{\succ}\\
		\sum\limits_{t^\succ\in\mathcal{T}_{\succ}}h(t^\succ)=1
	\end{align} 
\end{definition}

\begin{remark}
	Let $h:\mathcal{T}_\succ\longrightarrow [0,1]$ be a distribution over latent preference types. The distribution over response types implied by $h$ is the function $p_h:\Omega\longrightarrow[0,1]$ defined as:
	\begin{align}
		p_h(x_l,x_{l'}))=\sum\limits_{t^{\succ}\in N(x_l,x_{l'})}h(t^\succ)
	\end{align}	
	Where $N(x_l,x_{l'})=\{t^{\succ}\in\mathcal{T}_\succ|t^{\succ}(0)\in\mathcal{N}(x_l) \textit{ and } t^{\succ}(1)\in\mathcal{N}(x_{l'})\}$.
\end{remark}	

\begin{definition}
	A distribution over latent preference types $h$ is $\geq^*$-Submonotone if $h(t^\succ)=0$ for every $t^\succ$ that violates point 3 of assumption \ref{Assumption:SM}. This is, types $t^\succ$ such that $x_l\phantom{a}\geq^*x_{l'}$ and $x_l\succ_i^{t(z_0)} x_{l'}$ but  $x_{l'}\succ_i^{t(z_1)} x_{l}$ are not allowed.
\end{definition}	

\begin{assumption}[\textbf{No $S$ always takers}]
	Let $S\subset\mathcal{X}$ and $R\subset\mathcal{T}_Z$ a collection of instrument-response types . R has no $S$ always takers if $(x_l,x_l)\in R$ for all $x_l\in S$.
\end{assumption}

\begin{lemma}\label{Lemma:SubMonotonicity}
	Let $h$ be a distribution over latent preference types. Let $R$ be a collection of instrument-response types. Let $\geq^*$ be a binary relation (not necesarilly transitive or complete) over $\mathcal{X}$. Let $Ref_{\geq^*}=\{x_l\in\mathcal{X}|x_l \geq^* x_l\}$ denote the elements that belong to the reflexive part of $\geq^*$. Let $\geq^* _{Ref}$ and $\geq^* _{Irr}$ denote the reflexive and irreflexive part of $\geq^* $ respectively. Then:
	\begin{enumerate}
		\item $h$ satisfies $\geq^* _{Irr}$-Submonotonicity if and only if $h$ satisfies  $\geq^*$-Submonotonicity.
		\item If $\geq^*_{Irr}\setminus{\geq^R}\neq\emptyset$, then, a distribution over latent preference types $h$ exists such that $h$ satisfies ${\geq^R}$-Submonotonicity and $p_{h}(t^Z)=0$ for all $t^Z\in R$ but $h$ is not ${\geq^*}-$Submonotone.
		\item If ${\geq^R}\setminus{\geq^*}_{Irr}\neq\emptyset$, then, a distribution over latent latent preference types $h$ exists that is ${\geq^*}$-Submonotone but $p_h(t^Z)>0$ for some $t^Z\in R$. 
		\item If ${\geq^R}_{Ref}\neq\emptyset$ a distribution over latent response types exists that is ${\geq^R}$-Submonotone but $p_h(t^Z)>0$ for some $t^Z\in R$.
		\item If ${\geq^R}$-Submonotonicity holds and no $R_{\geq^M}$ always takers holds, then, $p(t^Z)=0$ for all $t^Z\not\in R$.
		\item If a distribution over response types $p$ is consistent with observed treatment probabilities and $p(t^Z)=0$ for all $t^Z\in R$ exists, then, a distribution $h$ over latent preference types exists such that $h$ is ${\geq^R}$-Submonotone and $p=p_h$. This last equality implies that $h$ is consistent with the data and that if $p$ satisfies no $S$ always takers, so does $p_h$. 
	\end{enumerate}
\end{lemma}

\begin{proof}[Proof of lemma \ref{Lemma:SubMonotonicity}]\phantom{a}\\
	\textbf{Part 1:}\\
	Suppose $x_l\geq^* x_{l'}$ and $x_l\succ_1^{(0)}x_{l'}$. If $x_l=x_{l'}$, then $x_l\succ_1^{(1)}x_{l'}$ holds trivially, otherwise, $\geq^*_{Irr}-QM$ implies that $x_l\succ_1^{(1)}x_{l'}$. The converse is trivial since $\geq^*_{Irr}\subset \geq^*$.\\
	\textbf{Part 2:}\\
	Take the ordering induced by the indexing of the elements of $\mathcal{X}$. Let $\succ^{x_l,x_{l'},*}$ be the preference relation that ranks $x_l$ as the best alternative, $x_{l'}$ as the second best, and all remaining alternatives according to their index: $x_l\succ^{x_l,x_{l'},*}x_{l'}\succ^{x_l,x_{l'},*}x_1\succ^{x_l,x_{l'},*}...\succ^{x_l,x_{l'},*}x_{L-1}$.\\
	\\
	Suppose that $x_l\neq x_{l'}$ and $(x_l,x_{l'})\in\geq^*_{Irr}\setminus\geq^R$. Let $t^\succ$ be the unique latent preference type that satisfies $t^\succ(0)=\succ^{x_l,x_{l'},*}$ and $t^\succ(1)=\succ^{x_{l'},x_l,*}$. Define $h$ as:
	\begin{align}
		h(t^\succ)=1&\\
		h(t^{\succ'})=0&\phantom{aaaaaa}\textit{ for all }t^{\succ'}\neq t^\succ
	\end{align}	
	Note that $(x_l,x_{l'})\not\in\geq^R$ implies that $h$ is $\geq^R-$Submonotone, but $h$ is not $\geq^*-$Submonotone (since preference reversals between $x_l$ and $x_{l'}$ occur with positive probability). Moreover, note that $p_h$ assigns zero probability to response types outside of $R$.
	\\
	\textbf{Part 3:}\\
	If $Ref_{\geq^R}$ is non-empty, a reasoning analogous to part 1 proves the claim. Otherwise, pick some $x_l\in Ref_{\geq^*}$ and define $h$ as follows:
	\begin{align}
		h((\succ^{x_{l},*},\succ^{x_{l},*}))=1&\\
		h(t^Z)=0 &\phantom{aaaaaa}\textit{ for all } t^Z\neq(\succ^{x_{l},*},\succ^{x_{l},*})
	\end{align}	
	$h$ is $\geq^*$-QM but $p_h$ assigns positive probability to the response type $(x_l,x_l)\not\in R$.\\
	\textbf{Part 4}\\
	This part follows immediately from part 3 by setting $\geq^*_{Irr}=\emptyset$.\\
	\textbf{Part 5:}\\
	Let $R'$ be the set of response types allowed by $\geq^{R}$-QM and no stayers at $Ref_{\geq^R}$. The goal is to show that $R\subset R'$.\\
	\\
	Suppose that $(x_l,x_{l'})\not\in R$. This implies that $x_l \geq^R x_{l'}$.\\
	\\
	If $x_l\neq x_{l'}$, consider a decision maker $i$ such that $X_i(0)= x_{l'}$, then, rationality implies that $x_{l'}=X_i(0)\succ_i^{(0)} x_{l}$. Therefore, by $\geq^R$-Submonotonicity, it must be that $x_{l'}\succ_i^{(1)}x_{l}$. This implies that $x_{l'}$ is not chosen by decision maker $i$ when $Z=1$ a.s.. Therefore $(x_{l},x_{l'})\not\in R'$.\\
	\\
	If $x_{l'}= x_{l}$, then, no stayers at $Ref_{\geq^R}$ implies that $(x_{l},x_{l'})\not \in R'$.\\
	\\
	Therefore, $t^Z\not\in R\Longrightarrow t^Z\not\in R'$ which is equivalent to  $t^Z \in\R'\Longrightarrow t^Z\in R$.\\
	\textbf{Part 6:}\\
	Let $t^\succ_{x_l,x_{l'},*}$ be the latent preference type such that $t^\succ_{x_l,x_{l'},*}(0)=\succ^{x_l,x_l{'},*}$ and $t^\succ_{x_l,x_{l'},*}(1)=\succ^{x_{l'},x_l,*}$. Define $h$ as follows:
	\begin{align}
		h(t^\succ)=\begin{cases} p(x_l,x_{l'}) &\textit{ if } t^\succ=t^\succ_{x_l,x_{l'},*}\\
			0 & otherwise
		\end{cases}
	\end{align}
	Note that $h$ is $\geq^R$-QM and that  for all $(x_l,x_l')\in R_Z$:
	\begin{align}
		p_h((x_l,x_l'))=\sum\limits_{t^\succ \in N(x_l,x_l')}h(t^\succ)=h(t^\succ_{x_l,x_{l'},*})=p(x_l,x_{l'})
	\end{align}
	
\end{proof}

Now I comment Lemma 1. The first part shows that only the irreflexive part of binary relations is relevant for submonotonicity. Part 2 shows that when a binary relation $\geq^*$ is stronger than the one induced by $R$, restricting response types to $R$ does not guarantee $\geq^*$-submonotonicity. Parts 3 and 4 show the converse, that imposing weaker than than $\geq^R$ dos not guarantee that response types outside of $R$ are ruled out. Part 5 shows that $R$ can be guaranteed by imposing $\geq^R$-submonotonicity and ruling out a specific set of always takers. 6 shows that $\geq^R$-submonotonicity is only falsified if the response type restriction associated to $R$ is falsified. Thus, in the binary instrument case $\geq^R$-submonotonicity  is just an interpretation for response type restrictions but yields no additional identifying power.  

Theorem \ref{Theorem:Submonotonicity} follows immediately from lemma 6 by noting there is a unique correspondence between $R$ and $\geq^R$ and therefore their cardinalities.

\subsection{Proofs for section \ref{Section:ManyIV}}

\begin{proofprop}[Proof of proposition \ref{Proposition:UpperBoundOnFT}]\phantom{a}\\
	
	First observe that for every $A\in\mathcal{B}$, every non-empty $\tilde{Z}\subset\mathcal{Z}$ and $x_l\in\mathcal{X}$, if $z_k\in\tilde{Z}$, then:
	\begin{align}
		\mathbb{P}[Y_i\in B_Y,X_i=x_l|Z_i=z_k]&\geq \mathbb{P}[Y_i(x_l)\in B_Y,X_i(\tilde{z})=x_l\phantom{a}\forall \tilde{z}\in\tilde{Z}|Z_i=z_k]\\
		\Longrightarrow\phantom{aaaaaa}\mathbb{P}[Y_i\in B_Y,X_i=x_l|Z_i=z_k]&\geq \mathbb{P}[Y_i(x_l)\in B_Y,X_i(\tilde{z})=x_l\phantom{a}\forall \tilde{z}\in\tilde{Z}]\\
		\Longrightarrow \underset{z_k\in\tilde{Z}}{\textit{min }}\big\{\mathbb{P}[Y_i \in B_Y,X_i=x_l|Z_i=z_k]\big\}&\geq \mathbb{P}[Y_i(x_l)\in B_Y,X_i(\tilde{z})=x_l\phantom{a}\forall \tilde{z}\in\tilde{Z}]\\
		\Longrightarrow\phantom{aaaaaaaaa}\underset{z_k\in\tilde{Z}}{min}\big\{\int\limits_{A}\phi_{Y|x_l,z_k}(y)d\mu_y\big\}&\geq\mathbb{P}[Y_i(x_l)\in B_Y,X_i(0)=x,X_i(\tilde{z})=x_l\phantom{a}\forall \tilde{z}\in\tilde{Z}]\label{Equation:CapConstraint}
	\end{align}	
	
	The initial inequality is true because outcomes in $B_Y$ might be realized by response types other than the IFF-takers associated to $\tilde{Z}$. The first implication follows from exogeneity of potential outcomes, the second is true because the inequality must hold for every individual $z_k\in\tilde{Z}$. The last simply applies the definition of $\phi_{Y|x_l,z_k}(y)$.
	
	Note that equation \ref{Equation:CapConstraint} must hold for every measurable set $B_Y\in\mathcal{B}_Y$. The next step is to transform these equations into a single restriction. 
	
	To this end, consider the set of all orderings over $\tilde{Z}$ denoted $\mathcal{O}(\tilde{Z})=\big\{\kappa:\{0,1,...,|\tilde{Z}|-1\}\longrightarrow\tilde{Z} : \exists \phantom{a}\kappa^{-1}\big\}$. For compactness, for every $\pi\in\mathcal{O}(\tilde{Z})$ and every $k\in\{0,1,...,|\tilde{Z}|-1\}$ I will denote the $k+1$-th element of $\tilde{Z}$ according to $\kappa$ as $\kappa_k$ (instead of $\kappa(k)$).
	
	Then, for every $x_l$ and $\kappa\in\mathcal{O}(\tilde{Z})$, define the set $Y_{x_l}^{\pi}$ as follows:
	\begin{align}
		Y_{x_l}^{\kappa}=&\Big\{y\in\mathcal{Y}: \phi_{Y|X=x_l,Z=\kappa_0}(y)\leq \underset{k=0,2,...,|\tilde{Z}|-1-1}{min} \{\phi_{Y|X=x_l,Z=\kappa_k}(y) \}\Big\}\\
		&\phantom{aaaaa}\bigcap \Big\{y\in\mathcal{Y}: \phi_{Y|X=x_l,Z=\kappa_1}(y)\leq \underset{k=1,...,|\tilde{Z}|-1-1}{min} \{\phi_{Y|X=x_l,Z=\kappa_k}(y) \}\Big\}\\
		&\phantom{aaaaa}...\bigcap \Big\{y\in\mathcal{Y}: \phi_{Y|X=x_l,Z=\kappa_m}(y)\leq \underset{k=m,...,|\tilde{Z}|-1}{min} \{\phi_{Y|X=x_l,Z=\kappa_k}(y) \}\Big\}\\
		&\phantom{aaaaa}...\bigcap \Big\{y\in\mathcal{Y}: \phi_{Y|X=x_l,Z=\kappa_{K-2}}(y)\leq \underset{k=|\tilde{Z}|-1-2,|\tilde{Z}|-1-1}{min} \{\phi_{Y|X=x_l,Z=\kappa_k}(y) \}\Big\}
	\end{align}
	This is, the set in which the values that the functions $\phi_{Y|x_l,z_k}(y)$ take are sorted by magnitude in exactly in the same order that $\kappa$ sorts the instrument values. Note that since the functions $\phi_{Y|x_l,z_k}(y)$ are point identified, the sets $Y_{x_l}^{\kappa}$ are  identified too. Moreover, the sets $Y_{x_l}^{\kappa}$ are the intersection of sets defined by inequalities that depend on measurable functions, thus, they are measurable themselves. 
	
	Then:
	\begin{align}
		\mathbb{P}[Y_i\in B_Y,X_i(\tilde{z})=x_l\phantom{a}\forall \tilde{z}\in\tilde{Z}|Z_i=z]&=\sum\limits_{\kappa\in\mathcal{O}(\tilde{Z})}\mathbb{P}[Y_i\in Y_{x_l}^{\kappa} \cap B_Y,X_i(\tilde{z})=x_l\phantom{a}\forall \tilde{z}\in\tilde{Z}|Z_i=z]\\
		&\leq\sum\limits_{\kappa\in\mathcal{O}(\tilde{Z})} \underset{z_k\in\tilde{Z}}{min}\big\{\int\limits_{Y_{x_l}^{\kappa}\cap B_Y}\phi_{Y|x_l,z_k}(y)d\mu_Y\big\}\\
		&=\sum\limits_{\kappa\in\mathcal{O}(\tilde{Z})} \int\limits_{Y_{x_l}^{\kappa}\cap B_Y}\phi_{Y|x_l,z_{\kappa_0}}(y)d\mu_Y\\
		&=\sum\limits_{\kappa\in\mathcal{O}(\tilde{Z})} \int\limits_{Y_{x_l}^{\kappa}\cap B_Y}\psi_{x_l,\tilde{Z}}(y)d\mu_Y
	\end{align}	
	
	The first line follows from the fact that the sets $Y_{x_l}^{\kappa}$ comprise a partition of $\mathcal{Y}$. The inequality in the second line follows from equation \ref{Equation:CapConstraint}. The equality in the third line is crucial, it is true because the sets $Y_{x_l}^{\kappa}$ are constructed so that $\phi_{Y|x_l,\kappa_0}(y)\leq \phi_{Y|x_l,z_k}(y)$ for every $y\in Y_{x_l}^{\kappa}$ and every $z_k\in\tilde{Z}$. Moreover, in $Y_{x_l}^{\kappa}$ it holds that $\psi_{x_l,\tilde{Z}}(y)=\phi_{Y|x_l,z_{\kappa_0}}(y)$ which proves the third line. Using again the fact that the sets $Y_{x_l}^{\kappa}$ partition $\mathcal{Y}$ yields the last equality.
\end{proofprop}

\section{Examples of instrument response-restrictions encompassed by the framework}\label{Appendix:RestrictionExamples}

\textbf{Multivalued ordered monotonicity:} Ordered monotonicity captures the idea that treatment values are strictly ordered (e.g, by intensity) and that as the instrument increases, units may only switch to weakly higher treatment levels. That is, for any pair of indices $k'>k>0$, response types such that $(t^Z)_k>(t^Z)_{k'}$ are ruled out. This restriction can be easily accommodated in the framework of assumption \ref{Assumption:SelectionModel} by imposing that response types as the one described before occur with zero probability. 

In their seminal paper about IV with heterogeneous treatment effects, \cite{angristImbens1994late} showed that if both the instrument and treatment are binary, and treatment is weakly monotonic in the instrument, the classical Wald estimand identifies a LATE for units that strictly increase their realized treatment when the instrument changes. This result has motivated the use of binary instruments and treatments in a wide range of empirical applications.

In a subsequent paper, \cite{angrist1995two} extended their analysis to multivalued ordered treatments. In this setting, no  LATE is point identified, but the Wald estimand continues to identify a convex combination of LATEs among units that are induced to take strictly higher treatments. 

Despite the interpretation provided by \cite{angrist1995two}, it is common practice in applied work to binarize monotone multivalued treatments. That is, the researcher selects a threshold $\bar{x}\in\mathcal{X}$ and defines a binary treatment indicator $X_i^{Bin}=\mathbf{1}_{\{X_i\geq\bar{x}\}}$, thereby reducing the setting to the binary treatment case. A justification for this practice is that the resulting Wald estimand has a simpler interpretation.

However, this practice has two limitations, first, there are instances in which the binarized model is rejected by a specification test, even though the correctly specified multivalued model is not. Second, after binarization, the Wald estimand identifies a weighted sum of LATEs for certain response types, but there is no guarantee that the weights add up to one. As a result, the estimand cannot generally be interpreted as LATE or even as a convex combination of LATEs. An in-depth analysis of the second limitation is provided in \cite{rose2024recoding}.

To illustrate the first limitation, suppose that $\mathcal{X}=\{x_0,x_1,x_2\}$, $\mathcal{Y}=\{0,1\}$ and $\mathcal{Z}=\{z_0,z_1\}$. Assume that half the population has the instrument-response type $(x_0,x_1)$, one fourth $(x_1,x_2)$, and the last fourth $(x_2,x_2)$%
\footnote{Recall that the instrument-response type $t=(x_l,x_{l'})$ represents units that realize treatment $x_l$ when $Z=z_0$ and $x_{l'}$ when $Z=z_1$.}.
Suppose that potential outcomes are homogeneous so that $Y_i(x_0)=0$, $Y_i(x_1)=1$ and $Y_i(x_2)=1$ for all $i\in\mathcal{I}$. Now suppose that a researcher sets $\bar{x}=x_2$. It is easy to see that units with  instrument-response type $(x_0,x_1)$ will not change their treatment status with respect to the artificially defined treatment variable $X^{Bin}$ as $Z$ increases, but, their potential outcomes will, violating the exclusion restriction in the binarized model.

Furthermore, \cite{kitagawa2015test} shows that if instrument validity holds in the binary case, the following equation must be satisfied:
\begin{align}
	\mathbb{P}[Y_i=1,X_i^{Bin}=0|Z=0]\geq\mathbb{P}[Y_i=1,X_i^{Bin}=0|Z=1]
\end{align}	
In the example, the left hand side is equal to $\frac{1}{4}$ while the right hand side is $\frac{1}{2}$, thus, instrument validity would be rejected for the binarized model, even if it holds in the true multivalued model with the right (multivalued ordered monotone) instrument-response restriction.

Furthermore, as studied by \cite{rose2024recoding}, even if equation \ref{Equation:BinaryExample} is not violated, the Wald estimand obtained using $Z$ as an instrument for $X^{Bin}$ may be difficult to interpret. To see this problem, suppose that the instrument-response types $(x_0,x_1)$, $(x_1,x_2)$, and $(x_2,x_2)$ occur with equal probability $\frac{1}{3}$ in the population. In this case, the Wald estimand becomes: 
\begin{align}
	W=\frac{\frac{1}{3}\mathbb{E}[Y_i(x_1)-Y_i(x_0)|i\in (x_0,x_1)]+\frac{1}{3}\mathbb{E}[Y_i(x_2)-Y_i(x_1)|i\in (x_1,x_2)]}{\frac{1}{3}}\label{Equation:BinaryExample}
\end{align}	
While the numerator is a weighted sum of LATEs, the weights do not necessarily sum to one. Consequently, the estimand $W$ cannot be interpreted as LATE or a convex combination of LATEs. Moreover, if instead of $\frac{1}{3}$, the denominator, i.e. the population probability associated with instrument-response type $(x_1,x_2)$, was small, the magnitude of the Wald estimand could be arbitrarily large. Thus, even if binarizing a monotone multivalued treatment may simplify the model, it can produce estimands that are both biased and difficult to interpret.

These concerns provide a strong motivation for retaining a correctly specified monotone model when the treatment is multivalued. 

Generalizing monotonicity to settings with multiple instruments is nontrivial. \cite{angrist1995two} proposed one approach, which imposes binary monotonicity in one direction between every pair of treatment values. This definition was later criticized by 
\cite{mogstad2021causal} who proposed an alternative extension. More recently, \cite{goff2020vector} introduced a tractable notion --vector monotonicity-- which ensures identification of a series of LATEs. 

Classical monotonicity and its three aforementioned generalizations can be expressed as restrictions that rule out specific response types, that is, they fall within assumption \ref{Assumption:SelectionModel}.  

\textbf{Unordered monotonicity and its refinements:} \cite{heckman2018unordered} proposed an alternative notion of monotonicity for settings in which no natural ordering over treatment values exists. Intuitively, unordered monotonicity requires that, as the instrument changes from one value to another, a subset of treatments $X^{I}\subset\mathcal{X}$ are equally incentivized relative to the treatments in $\mathcal{X}\setminus\mathcal{X}^I$. This restriction not only implies that take-up of the incentivized treatments must weakly increase but also rules out unit-level switches within the set of promoted treatments $\mathcal{X}^I$, or its complement $\mathcal{X}\setminus\mathcal{X}^I$. This is another restriction on heterogeneity obtained from ruling out particular response types and thus, it can be accommodated within the framework of this paper.

\cite{heckman2018unordered} argue that the identifying power of unordered monotonicity can be strengthened by further restricting the set of admissible response types. For example, suppose that $x_0\in\mathcal{X}^I$ represents a discounted car, and $x_1,x_2\in\mathcal{X}\setminus \mathcal{X}^I$ represent, respectively, a similar full-priced car and a pick-up truck. One could impose the restriction that only individuals who would have chosen $x_1$ in the absence of the discount may switch to $x_0$ as a result of the discount, ruling out switches from $x_2$ to $x_0$. Unordered monotonicity augmented with this type of restriction remains within the framework of this paper. 

\textbf{Other response type restrictions:} The identifying power of response type restrictions extends beyond ordered and unordered monotonicity. \cite{goff2024does} shows that the ability of restrictions on treatment-response heterogeneity to point indetify some LATE without restricting the outcome space is underpined by the set of treatment-response types it allows. Leveraging this insight, \cite{goff2024does} identifies a series of novel response-type restrictions that guarantee point identification of at least one LATE, and a procedure to find other such restrictions.

To advance the interpretability of IV models with multiple instruments and treatments, \cite{lee2020treatment} investigated identification when multiple instruments are available and each instrument targets (i.e. incentivizes) adoption of one particular treatment more than the rest. In particular, \cite{lee2020treatment} provide a general framework and then --to achieve point or partial identification-- introduce the notions of strict targeting, and one-to-one targeting. Multiple instances of these two assumptions fit into the framework studied in this paper.

In a closely related paper, \cite{bai2024sharp} study the case in which every treatment value is uniquely incentivized by a single instrument. For this particular targeting structure they show that multiple LATEs are point identified and derive the sharp testable implications of the model. However, the assumption that each instrument uniquely targets one --and only one-- treatment (without affecting any other comparison) might be overly restrictive in multiple empirical settings. The framework developed in this paper can be used to assess the validity of the complete model proposed by \cite{bai2024sharp}, as well as more plausible relaxations.

\textbf{Other linear restrictions on instrument-response types:} Rather than ruling out specific response types, a less restrictive approach to IV identification is to place bounds on the prevalence or relative frequency of certain behaviors or instrument-response types. A notable example is the ``more compliers than defiers'' condition of \cite{de2017tolerating}. Another example is \cite{richardson2010analysis}, who derive bounds on the outcome distribution under assumptions on the fraction of always takers in the population.

Similar assumptions are also used as intermediate steps in sensitivity analysis methods that assess the robustness of IV estimands to relaxations of monotonicity (\cite{huber2014sensitivity}, \cite{noack2021sensitivity}, \cite{yap2025sensitivity}). The tests from this paper can be used to falsify these types of restrictions on their own or combined with response-type restrictions.

\section{Examples of submonotonicity}\label{Appendix:Submonotonicity}

\textbf{Ordered monotonicity:} Ordered monotonicity captures the notion that an implicit ranking (a total order) of alternatives exists, and, as the instrument increases, adoption of higher ranked alternatives is weakly incentivized. This assumption is implicit in most random utility models in which a mean structural utility function is common to all members of the population and unit-specific unobservables are additively separable. Thus, holding unobserved heterogeneity fixed, as observables change all decision makers agree on which alternatives yield higher utility levels and thus are being promoted. 

Because subomotonicity allows for general binary relations, ordered monotonicity is a particular case.

\FloatBarrier
\begin{figure}[h!]
	\centering
	\includegraphics[width=0.5\linewidth]{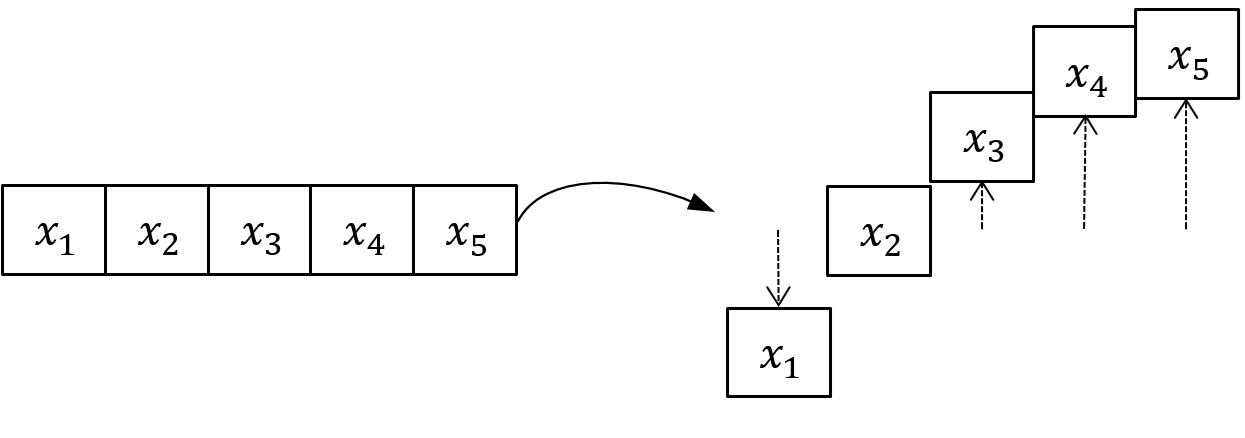}
	\caption{Ordered monotonicity}
	\footnotesize \raggedright \textit{\textbf{Notes:}} Visual representation of ordered monotonicity. The left of each figure represents alternatives before the intervention ($Z=z_0$). The right represents the effect of the intervention ($Z=z_1$). The higher an alternative is depicted on the left side, the more it is being incentivized. The black arrows represent the relative magnitude of the incentives.
\end{figure}
\FloatBarrier

\textbf{Unordered monotonicity:} Unordered monotonicity seeks to abstain from imposing an ex-ante ordering of alternatives and instead simply acknowledging that instruments incentivize adoption of some treatments relative to the rest. This idea is formalized by assuming that a subset of treatments $S\subset\mathcal{X}$ exists such that all alternatives in $S$ are being promoted and no alternative in its complement $S':=\mathcal{X}\setminus S$ is. 

While this assumption is general in spirit and in the multiple instruments setting unordered monotonicity is not implied nor implies ordered monotonicity, in the binary instrument case, it is a strictly more restrictive assumption. The reason is that unorderdered monotonicity not only requires that alternatives in $S$ are being promoted, but also that they are being promoted with exactly the same intensity so that switches within the set $S$ (and its complement) occur with zero probability. 

The previous point is easy to see by noting that the binary relation associated to unordered monotonicity is defined as follows:

\begin{align}
	x_l\geq^R x_{l'} \Longleftrightarrow \begin{cases}	 x_l,x_{l'}\in S & or \\
		x_l,x_{l'}\in S'& or   \\
		x_l\in S, x_{l'}\in S' \label{Equation:UOM}
	\end{cases}
\end{align}
This restriction induces a partial order with two indifference classes $S$ and $S'$, so that if $x_l,x_{l'}\in S$ then $x_l\geq ^Rx_{l'}$ and $x_{l'}\geq^R x_{l}$. The name partial order is misleading in this setting --it suggests weakening a total order--. The weakening is achieved by imposing that different treatments belong to the same equivalence class, that is, adding more ordered pairs to the binary relation which imposes additional restrictions on heterogeneity. 
\FloatBarrier
\begin{figure}[h!]
	\centering
	\includegraphics[width=0.5\linewidth]{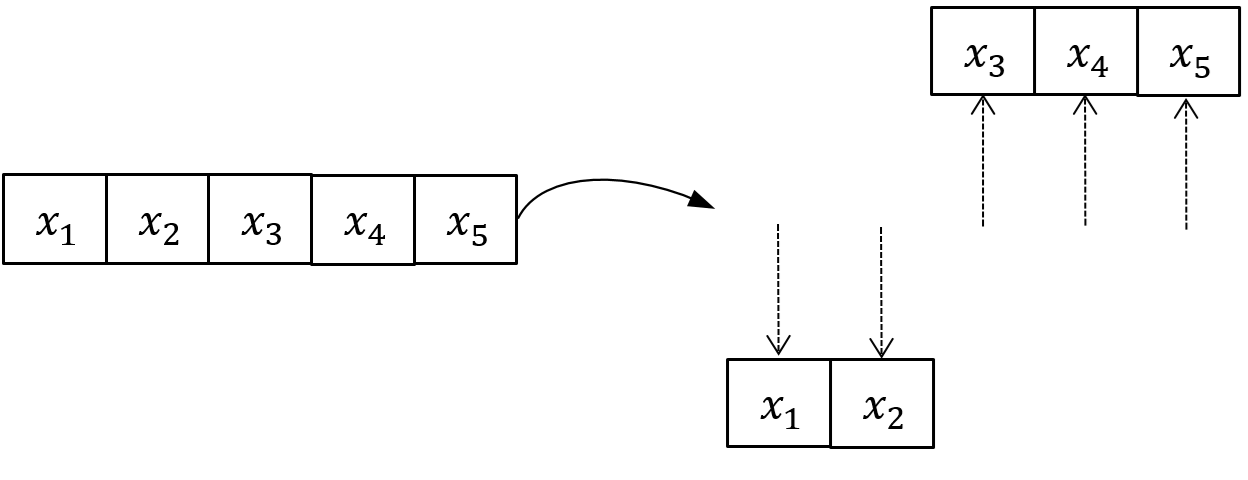}
	\caption{Unordered monotonicity}
	\footnotesize \raggedright \textit{\textbf{Notes:}} Visual representation of unordered monotonicity. The left of each figure represents alternatives before the intervention ($Z=z_0$). The right represents the effect of the intervention ($Z=z_1$). The higher an alternative is depicted, the more it is incentivized. The black arrows represent the relative magnitude of the incentives. The set $S=\{x_3,x_4,x_5\}$ represents the alternatives that are being incentivized relative to $S'=\{x_1,x_2\}$. Note that comparisons between alternatives within the same set are not altered.
\end{figure}
\FloatBarrier
Furthermore, note that if $S$ has only one element $s$ (i.e. $|S|=1$), then, an instrument that encourages adoption of $s$ (\cite{lee2020treatment}) or strictly targets $s$ (\cite{bai2024sharp}) satisfies a particular case of unordered monotonicity. Targeting or encouragement are not implied nor imply unordered monotonicity in the multiple instruments case. 

The sets $S$ and $S'$ could be further partitioned into three sets instead of two. For example, let $S=S_1\cap S_2$ and $S'=S_1'\cap S_2'$ and let alternatives in $S_1$ and $S'_1$ represent compact cars and alternatives in $S_2$ and $S_2'$ SUVs. Suppose that $S$ represents the set of all discounted vehicles (SUVs and compact) and $S'$ the set of vehicles without a discount. In this setting, one plausible restriction is that the discount on compact cars only affects decision makers who, absent the discount, would buy alternatives in $S_1'$; and the discount on SUVs only affects decision makers who would pick alternatives in $S_2'$. Thus, switches from $S_1'$ to $S_2$ or $S_2'$ to $S_1$ are assumed to occur with zero probability. This restriction is no longer represented by a partial order, but it is derived from adding ordered pairs to the partial order described in equation \ref{Equation:UOM}, thus, it is a submonotone restriction associated to the extension of a partial order.    

\textbf{Relaxations of unordered monotonicity:} The assumption that all alternatives are equally discounted might be too strong, instead, one could assume that there are $N$ discount bins and that comparisons within the same bin are not affected but comparisons across different bins are --always in favor of alternatives with larger discounts. For example, if $N=3$ there could be a set of heavily discounted cars, a set of cars with payment facilities, and a set of cars that have to be payed upfront. This situation also corresponds to submonotonicity with respect to a partial order but now with three indifference classes. 

Additional restrictions, as the one considered in the compact cars and SUVs examples could also be added across all three discount bins. The resulting restriction would remain within the submontonicity framework.

\FloatBarrier
\begin{figure}[h!]
	\centering
	\includegraphics[width=0.5\linewidth]{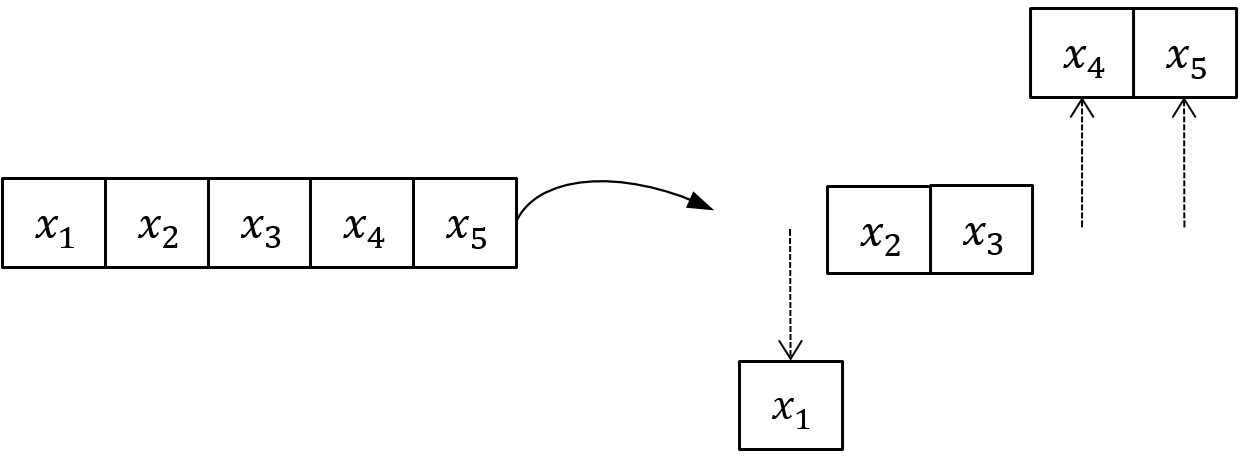}
	\caption{Submonotonicity with respect to a partial order with three indifference classes}
	\footnotesize \raggedright \textit{\textbf{Notes:}} Visual representation of submonotonicity with respect to a partial order with three indifference classes. The left of each figure represents alternatives before the intervention ($Z=z_0$). The right represents the effect of the intervention ($Z=z_1$). The higher an alternative is depicted, the more it is incentivized. The black arrows represent the relative magnitude of the incentives. The set $\{x_4,x_5\}$ represents the alternatives that are being incentivized the most (heavily discounted), the set $\{x_2,x_3\}$ the alternatives that for which payment facilities are offered (i.e. weakly incentivized). Alternative $x_1$ is not being incentivized in any way. Note that comparisons between alternatives within the same set are not altered and thus, their difference in height remains equal.
\end{figure}
\FloatBarrier
\textbf{Relaxations of encouragement:} Now consider an intervention in which an instrument encourages adoption of certain treatments in $S\subset\mathcal{X}$, but there is a concern that the instrument might also affect utility from other treatments. For example, consider an an intervention aimed at increasing sports participation by providing information. It is clear that such an intervention could be used as an instrument for sports practice, but the intervention may also promote other activities that develop motor or social skills such as community service or playing board-games. Hence, the restriction that the instrument is not affecting comparisons between treatments other than practicing sports might not be plausible and a weaker assumption might be preferred. 

One such weaker assumption could be that alternatives $x_1$ and $x_2$ which represent sports are certainly being promoted relative to all others (since they are the main target of the intervention), but pairwise comparisons between other treatments might also affected and the direction of this effect is heterogeneous across the population. Thus, restrictions on the direction of the flows between non-sports activities cannot be credibly imposed. This situation is encoded by submonotonicity with respect to a quasi-order, i.e. a binary relation that is transitive but not complete. 

If $x_5$ and $x_4$ represent different sports and $x_3$, $x_2$ and $x_1$ other activities, the specific submonotonicity assumption would be $x_l\geq^R x_{l'}$ for $l=5,4$ and $l'=3,2,1$ but neither $x_l\geq^R x_{l'}$ nor $x_{l'}\geq ^R x_l$ for $l,l'=5,4$ or $l,l'=3,2,1$. A graphical representation is given in figure \ref{Figure:UOM}. 

\FloatBarrier
\begin{figure}[h!]
	\centering
	\includegraphics[width=0.5\linewidth]{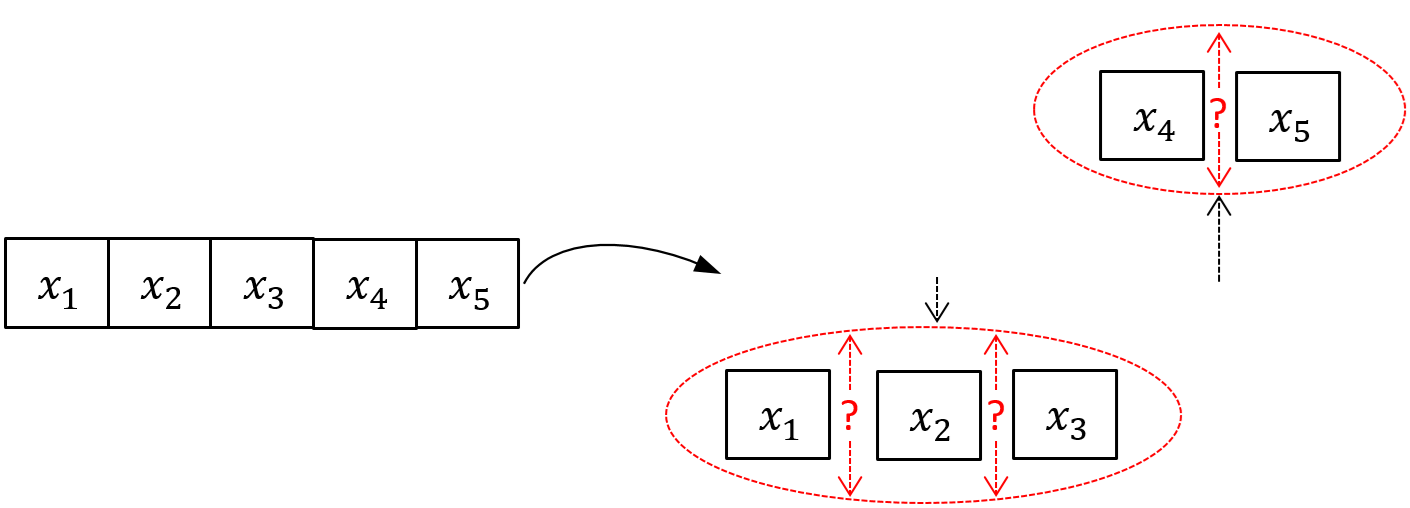}
	\caption{Submonotonicity with respect to a  quasi-order - Sports example }
	\label{Figure:UOM}
	\footnotesize \raggedright \textit{\textbf{Notes:}} Visual representation of submonotonicity with respect to a quasi-order. The left of each figure represents alternatives before the intervention ($Z=z_0$). The right represents the effect of the intervention ($Z=z_1$). The higher an alternative is depicted on the right side, the more it is incentivized. The set $\{x_4,x_5\}$ represents the alternatives that are being incentivized the most (sports), the set $\{x_1,x_2,x_3\}$ the alternatives that are not being incentivized. The black arrows represent the relative magnitude of the incentives but the red arrows represent relative uncertainty about this effect within the set of circled alternatives. That is, within the red elipses no restriction on how pairwise comparisons are affected is imposed.
\end{figure}
\FloatBarrier

Note that relaxing the completeness requirement allows for more heterogeneity. In particular, it allows decision makers to change their relative ranking of the same pair of alternatives in different directions. This situation cannot be captured by a random utility model in which the effect of $Z$ is common to al decision makers, instead, it is allowing the effect of $Z$ to be determined by unobserved unit characteristics in the spirit of a random coefficients model.  

Now consider an intervention aimed at preventing students from dropping college to take part-time jobs.  Suppose this goal is achieved by providing additional classes aimed at developing professional qualifications but also by providing information on the importance of college completion. 

It is clear that such an intervention would promote staying in college relative to taking part-time jobs, however, the classes could help students find full-time jobs, even without the degree. On the other hand, the information could reduce the willingness of students to accept any type of job. If the latter effect is stronger for full-time jobs (because they are more demanding and thus increase the risk of not finishing college), some students might choose to leave a full-time job and take a part-time one. Thus, it might be desirable to only impose the restriction that college is being incentivized relative to part-time jobs, but nothing else.

This restriction can also be accomodated by submonotonicity by assuming that college degrees $x_4$ and $x_5$ are being incentivized relative to half-time jobs $x_3$ and $x_2$ but no restrictions involving full-time jobs $x_1$ are imposed. Formally, $x_l\geq^R x_{l'}$ for $l=1,2$ and $l'=3,4$.   

\FloatBarrier
\begin{figure}[h!]
	\centering
	\includegraphics[width=0.5\linewidth]{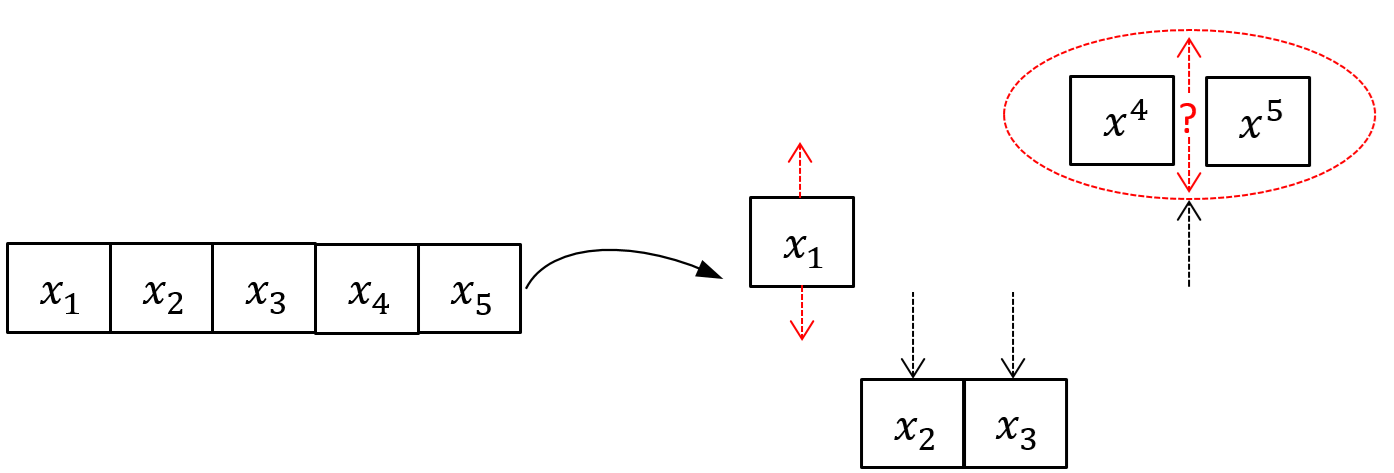}
	\caption{Submonotonicity with respect to a transitive quasi-order - College and full-time jobs example. }
	\footnotesize \raggedright \textit{\textbf{Notes:}} Visual representation of submonotonicity with respect to a quasi-order. The left of each figure represents alternatives before the intervention ($Z=0$). The right represents the effect of the intervention ($Z=1$). The higher an alternative is depicted, the more it is incentivized. The set $\{x_4,x_5\}$ represents the alternatives that are being incentivized the most (sports), the set $\{x_1,x_2,x_3\}$ the alternatives that are not being incentivized. The black arrows represent the relative magnitude of the incentives but the red arrows represent relative uncertainty about this effect relative to alternatives in the same class. Thus, within the red circle, no stance is taken on how pairwise comparisons are affected. Furthermore, no stance on whether $x_1$ is being promoted or not relative to $x_4$, $x_5$ or even $x_2$ and $x_3$ is taken.
\end{figure}
\FloatBarrier

\textbf{Relaxing transitivity/encouraging experimentation:} The submonotonicity framework does not require transitivity of the binary relation $\geq^R$ --i.e. it allows for cycles in $\geq^R$. This observation raises the question of how can such a restriction be interpreted --specially in light of \cite{goff2024does} who shows that some models associated to non-transitive binary relations can point identify LATEs for some specific response groups.  

To answer this question, consider a three alternatives example and suppose that $x_1>^Rx_2$, $x_2>^Rx_3$ and $x_3>^Rx_1$ (where $x_l>^Rx_{l'}$ indicates $x_l\geq ^Rx_{l'}$ but not $x_{l'}\geq^Rx_{l}$).

This restriction implies that:
\begin{itemize}
	\item Any decision maker who choses $x_1$ when $Z=0$ will now choose $x_1$ or $x_3$.
	\item Any decision maker who choses $x_2$ when $Z=0$ will now choose $x_2$ or $x_1$.
	\item Any decision maker who choses $x_3$ when $Z=0$ will now choose $x_3$ or $x_2$.
\end{itemize}
Furthermore, if the assumption that there are no always takers at $\mathcal{X}$ is added (i.e. all types of always takers are ruled out), then: 
\begin{itemize}
	\item Any decision maker who choses $x_1$ when $Z=0$ will now choose $x_3$.
	\item Any decision maker who choses $x_2$ when $Z=0$ will now choose $x_1$.
	\item Any decision maker who choses $x_3$ when $Z=0$ will now choose $x_2$.
\end{itemize}
In the first case, choices when $Z=0$ restrict potential choices when $Z=1$. In the second, there is a one to one correspondence between one and the other. In both cases, if decision makers change their choice, they pick an alternative different than the one they would have chosen absent the instrument. The result is a situation where every alternative is promoted among subjects who do not choose it at the initial value of the instrument; and discouraged among those that already prefer it over the rest. The previous discussion suggests interpreting the effect of the instrument as promoting experimentation or variety. 

Interpreting more general non-transitive restrictions is not straightforward, but, in all cases, submonotonicity induces restrictions on heterogeneity that can be expressed as a correspondence of choices when $Z=0$, and thus, can be seen as restricting heterogeneity in terms of status-quo choices.

\end{document}